\documentclass{article}

\usepackage{lscape}
\usepackage{amsmath}
\usepackage{amsthm}
\usepackage{amssymb}
\usepackage{latexsym}
\usepackage{xcolor}

\usepackage{proof}
\usepackage{bussproofs}
\usepackage{amssymb}
\usepackage{amsmath}
\usepackage{mathrsfs}
\usepackage{color}
\usepackage{cmll}
\usepackage{bbold}
\usepackage{tikz}
\usepackage{tikz-qtree}
\usepackage{stmaryrd}

\usepackage{hyperref}

\newtheorem{theorem}{Theorem}[section]
\newtheorem{lemma}[theorem]{Lemma}
\newtheorem*{theorem*}{Theorem}
\newtheorem{proposition}[theorem]{Proposition}
\newtheorem*{corollary}{Corollary}

\theoremstyle{definition}
\newtheorem{definition}{Definition}[section]

\theoremstyle{remark}

\newcommand{\Impl}{\Rightarrow}
\newcommand{\impl}{\rightarrow}
\newcommand{\et}{\wedge}

\newcommand{\all}{\forall}

\newcommand{\vphi }{\varphi}

\newcommand{\cho}[1]{\langle#1)}
\newcommand{\chot}{\oplus}

\newcommand{\efq}{\mathrm{efq}}

\newcommand{\thr}{\varepsilon}

\newcommand{\lam}{\lambda}
\newcommand{\lan}{\langle}
\newcommand{\ran}{\rangle}

\newcommand{\trust}{\mathrm{Trust}}
\newcommand{\tconst}{\mathrm{T}}

\newcommand{\calc}{\mathcal{C}}
\newcommand{\cald}{\mathcal{D}}

\newcommand{\calt}{\mathcal{T}}

\newcommand{\boldn}{\mathbb{N}}
\newcommand{\boldq}{\mathbb{Q}}
\newcommand{\bolds}{\mathbb{S}}

\newcommand{\red}{\mathrm{Red}}

\newcommand{\cruno}{\mathrm{CSN}}
\newcommand{\crdue}{\mathrm{CEval}}
\newcommand{\crtre}{\mathrm{CAntiEval}}

\providecommand{\keywords}[1]
{
  \small	
  \textbf{\textit{Keywords.}} #1
}

\title{A Strongly Normalising\\System of Dependent Types\\for Transparent and Opaque\\Probabilistic Computation}

\author{Francesco A. Genco\thanks{Work funded by the PRIN project n.2020SSKZ7R BRIO (Bias, Risk and Opacity in AI).}\\
{\small LLC, Department of Philosophy and Education Sciences, University of Turin}}

\begin{document}

\maketitle

%\begin{abstract}
%The extensive deployment of probabilistic algorithms has radically changed our perspective on several well-established computational notions. {\it Correctness} is probably the most basic one. While a typical probabilistic program cannot be said to compute {\it the} correct result, we often have quite strong expectations about the frequency with which it should return certain outputs. In these cases, {\it trust} as a generalisation of correctness fares better. One way to understand it is to say that a probabilistic computational process is trustworthy if the frequency of its outputs is compliant with a probability distribution which models its expected behaviour. We present a formal computational framework that formalises this idea. In order to do so, we define the typed $\lam$-calculus \black{} which features operators for conducting experiments at runtime on probabilistic programs and for evaluating whether they compute outputs as determined by a target probability distribution. After proving some fundamental  computational properties of the calculus, such as progress and termination, we define a static notion of {\it confidence} that enables us to prove that our notion of {\it trust} behaves correctly with respect to the basic tenets of probability theory.
%
%\end{abstract}
\keywords{%trust; 
probabilistic computation; non-deterministic computation; opacity; type theory; dependent types}

\section{Introduction}

We define an extension of  $\lam$-calculus with dependents types that enables us to encode transparent and opaque probabilistic programs and prove a strong normalisation result for it by a {\it reducibility} technique. While transparent non-deterministic programs are formalised by rather usual techniques, opaque non-deterministic programs are formalised by introducing in the syntax oracle constants, the behaviour of which is governed by oracular functions. The generality of these functions and the fact that their values are determined by the form of the whole term inside which the relative  oracle occurs also enable us to simulate learning-like behaviours. 

We then extend the calculus in order to define a computational trustworthiness predicate.
%in the style of \cite{gp23}. The novelty of this approach with respect to  \cite{gp23} lies in the fact that 
The extension of the  calculus does not only enable us to precisely formalise a notion of trustworthiness and to encode the procedures required to test it on programs, but also to reason, by means of the type system, on the behaviour of programs with respect to trustworthiness.

\section{Kinds, types and terms}

\begin{definition}[Typeless Terms]\label{def:terms} The
terms of the calculus are defined by the following grammar:
\begin{align*}x,y,z\ldots  \;\; ::=  \;\; & x_0 \;\; \mid\;\; x_1 \;\; \mid\;\; x_2\;\; \ldots  
\\
o  \;\; ::=  \;\; & o_0 \;\; \mid\;\; o_1 \;\; \mid\;\; o_2\;\; \ldots    
\\
t,s,u,v\ldots  \;\; ::= \;\; & x \;\; \mid\;\; o  \;\; \mid\;\; ot  \;\; \mid\;\;   o\nu  \;\; \mid\;\; (ot)\nu  \;\; \mid \\ &   \lambda x:\vphi .t \;\; \mid\;\;  ts  \;\; \mid \\ & t  \cho{p} s\;\; \mid\;\; t\nu  \;\; \mid
\\
&\langle t,s \ran  \;\; \mid\;\; t\pi_j   \;\; \mid
 \;\; [\kappa  ]
   \;\; \mid\;\; 
%   [\kappa ]^p  \;\; \mid \;\;
 [t, [\kappa_1 /\dots/\kappa _n] , s]
% \;\;\mid \;\; [t, [\kappa_1 /\dots/\kappa _n] , s]^p\;\;\mid
%\\
%& [t, [\kappa_1  /\dots/\kappa _n] , s]^p u
%  \;\;\mid\;\;[[\kappa_1 /\dots/\kappa _n] , s]\;\;\mid \;\; [[\kappa_1 /\dots/\kappa _n] , s]^p\;\;\mid
%\\
%& [[\kappa_1  /\dots/\kappa _n] , s]^pu
%\\
%\kappa  \;\; ::= \;\; & \varnothing \;\;\mid\;\;t _1 , \dots , t_n
\end{align*}
where $n\in \boldn$, $j\in \{0,1\}$ and $p\in \{i\mid i\in \boldq,  0\leq i\leq 1\}$\end{definition}

\begin{definition}[Kindless Type constructors]\label{def:types} The
terms of the calculus are defined by the following grammar:
\begin{align*}\alpha  \;\; ::=  \;\; & \alpha _0 \;\; \mid\;\; \alpha _1 \;\; \mid\;\; \alpha _2\;\; \ldots
\\
\vphi   \;\; ::= \;\; & \alpha  \;\; \mid\;\;  \lambda  x :\vphi _1 . \vphi _2   \;\; \mid\;\; \vphi  t    \;\; \mid\;\;   \all  x :\vphi _1 . \vphi _2  \;\; \mid\;\; \chot\vphi \;\; \mid\;\; \Sigma \vphi 
 \;\; \mid\;\;
\vphi _1 \et\vphi _2\end{align*}\end{definition}

The standard part of the grammar for terms includes variables $x_0,x_1\dots $ for terms, applications and $\lambda$-abstractions. Nondeterministic choice terms of the form $t\cho{p}s$ are supposed to represent programs that nondeterministically reduce to either $t$, with probability $p$, or to $s$, with probability $1-p$. 
%Traditional tuples $\langle t _1 , \dots , t_n \ran$ are also terms of our language, but, as we will see in the typing rules, are only used to collect the results of tests on the probabilistic behaviour of  oracle constants. 
The terms of the form $[\kappa ]$, where $\kappa$ is a possibly empty list $t _1 , \dots , t_n$ of terms, denote computations represented as reductions of $t_i$ in $t_{i+1}$ for $1\leq i\leq n $. If a term representing a computation is of the form $[\kappa]^p$, then it also indicates that the computation has probability $p$ of developing as specified by the list $\kappa$ of terms. 
%A computation application $[\kappa]^pt$ (or $[t, [\kappa_1  /\dots/\kappa _n] , s]^pu$) denotes the operation of letting the computation $[\kappa]$ (or $[t, [\kappa_1  /\dots/\kappa _n] , s]$)  act on the term $s$ (or $u$). This kind of action simply checks whether the argument $s$ (or $u$) is a suitable starting point for the applied computation.

Instead of defining a grammar of types directly, we define, as per standard practice, a grammar of type constructors, that is, syntactic objects that are supposed to evaluate to types. Thus we can handle smoothly the definition of types that depend on terms.
According to the grammar, a type constructor can be a type constructor variable $\alpha_1, \alpha _2\dots $, a $\lambda$-abstraction of a type constructor---where the bound variable is a variable for terms---or the application of a type constructor to a term. Type constructor applications are supposed to be evaluated in order to reduce the type constructor to a proper type, which is a type constructor that does not contain occurrences of $\lambda$ but only occurrences of variables $\alpha_1, \alpha_2\dots $ and of $\all$  quantifiers. When we form a type constructor, we can also generalise a type constructor by the $\all$ quantifier---where the bound variable is a variable for terms. As opposed to the occurrences of $\lambda$, the occurrences of these quantifiers are not supposed to disappear during the evaluation of the type constructor. Indeed, $\lambda $ in a type constructor of the form $(\lambda x:\vphi _1 . \all y:\vphi  _1 \alpha xy)t$ is meant to have  the resulting type depend on the term $t$ given as an argument to the type constructor, while  $\all $ is meant to stay as a first-order quantifier occurring also in the type resulting from the evaluation of the constructor: $(\lambda x:\vphi _1 . \all y:\vphi  _1 \alpha xy)t$ evaluates to $\all y:\vphi  _1 \alpha ty$, which is a type. Notice that the variables $\alpha, \alpha_1, \alpha_2, \dots $ are meant to represent first-order predicates. That is, what we would normally write as $\all x.\all y.P(x,y)$ is here written as $\all x:\vphi  _1.\all y:\vphi  _2.\alpha xy $ where $\alpha $ is meant as a type constructor variable corresponding to the predicate symbol $P$. As specified by the deductive system below, the kind of $\alpha$ will determine its arity.

Type constructors encode functions just as terms do, with the only difference that terms take terms as inputs and yield terms as outputs while type constructors take terms as inputs and yield types as outputs. Therefore, in order to guarantee that their rewriting system is well-behaved, we introduce a system to constrain their usage and obtain a terminating system. Just as we assign types to terms, we assign kinds to type constructors.  

The grammar of kinds is rather simple: we have a base case, the constant $*$ denoting the kind of types, and an abstraction $\Pi$ that binds term variables.\footnote{Notice that $x$ can never occur as a sub-kind of a kind $\Pi x:\vphi  . \Phi$. Hence, quantified kinds of dependent type constructors can be seen as implication kinds of the form $\vphi  \Impl \Phi$ (similarly to what we will do for types $\all x:\vphi _1 . \vphi _2 $ where $x$ does not occur in $\vphi _2$) . For instance, all the information provided in the kind of the the type constructor $(\all x:\vphi . \alpha x): (\Pi x:\vphi .* )$ is also contained in the kind $\vphi  \Rightarrow * $ since $x$ does not occur in $*$.} For instance

\begin{definition}[Kinds]\label{def:kinds} The
kinds of the calculus are defined by the following grammar:
\begin{align*}\Phi   \;\; ::=  \;\; & * \;\; \mid\;\; \Pi x : \vphi  . \Phi 
\end{align*}
\end{definition}

The notions of {\it free variables}, {\it bound variables} and {\it substitution} are defined as usual, see, for instance, \cite{sor06}.

%{\option{\latent  
%Now, the idea behind the use of types $A\chot B $ for probabilistic programs is quite simple: a program of type $A\chot B $ can either reduce to a term of type $A$ or to a term of type $B$. This simplicity must be paid in complications when it comes to type preservation. Quite simply, type preservation does not hold, in a strict sense. What we can prove about reduction, though, is that a term of type $A$ will always reduce to a more specific term. For all practical purposes this is enough since all processes that accept something of type $A$ as an argument will also accept something with a more specific type as an argument: one can always treat an element of a set as an element of any of its supersets. In order to formalise a notion of type specificity, we need to define the notion of subtype. }
%}

%\begin{definition}[Subtype]\label{def:subtype}
%Given a relation $\cals$ defined on atomic types, the subtyping relation $<:$  extending $\cals$ to complex types is  defined as follows:
%\begin{itemize}
%\item for any type $A$, $A<:A$
%\item for any three types $A,B,C$, if $A<:B$ and $B<:C$, then $A<:C$
%\item if $A$ and $B$ are atomic types and $A\cals B$, then $A<:B$ 
%\item if $B_1<:A_1$ and $A_2<:B_2$, then $\all x:A_1. A_2<:\all y:B_1. B_2$
%\item if for any $A_j$ there is a $B_i$ with $i,j \in \{1, 2\}$ such that $A_j<:B_i$, then $A_1\chot A_2  <: B_1\chot B_2$
%%\item for any oracle constant $o$, term $t$ and natural numbers $m,n\in \boldn$, $o\Mapsto _{m/n}t<:o\Mapsto _{?/n}t$
%\end{itemize}
%\end{definition} 

\begin{definition}[Variables occurring in a kind] The term variables occurring in a kind $\Phi$ are all the term variables that appear in $\Phi$ but not immediately to the right of a $\Pi$.
\end{definition}

\begin{definition}[Subconstructors and variables occurring in a type constructor]
\label{def:subconstructor}
The subconstructors of a type constructor are inductively defined as follows. 
\begin{itemize}
\item $\alpha x_1 \dots x_n$ (possibly for $n=0$)  is a subconstructor of $\alpha x_1 \dots x_n $.

\item $\lambda  x :\vphi _1 . \vphi _2$ and all subconstructors of $\vphi _2$ 
are subconstructors of $\lambda  x :\vphi _1 . \vphi _2$.

\item $\vphi  t$ and all subconstructors of $\vphi $ 
are subconstructors of $\vphi  t$.

\item $\all  x :\vphi _1 . \vphi _2  $ and all subconstructors of $\vphi _2   $ 
are subconstructors of $\all  x :\vphi _1 . \vphi _2  $.

\item $\chot\vphi$ and all subconstructors of $\vphi $
are subconstructors of $\chot\vphi$.

\item $\Sigma \vphi $ and all  subconstructors of $\vphi $ 
are subconstructors of $\Sigma \vphi $.
\end{itemize}
The constructor variables occurring in a type constructor $\vphi$ are all the constructor variables that appear in $\vphi$.

The term variables occurring in a type constructor $\vphi$ are all the term variables that appear in $\vphi$ but not immediately to the right of a $\lambda$ or $\all$.
\end{definition}

\begin{definition}[Subterms and variables occurring in a term]\label{def:subterm}
The subterms of a term are inductively defined as follows.
\begin{itemize}
\item $x$ is a subterm of $x$, $o$ is a subterm of $o$.

\item $\lambda x:\vphi .t $ and all subterms of $t$ are subconstructors of $\lambda x:\vphi .t$.

\item $ts $ and all subterms of $t$ and $s $ are subconstructors of $ts$.

\item $t  \cho{p} s$ and all subterms of $t$ and $s$ are subconstructors of $t  \cho{p} s$.

\item $t\nu$ and all subterms of $t $ are subconstructors of $t\nu$.

\item $\langle t,s \ran$ and all subterms of $t$ and $s$  are subconstructors of $\langle t,s \ran$.

\item $t\pi_j$ and all subterms of $t$ are subconstructors of $t\pi_j$.

\item $[\kappa ]^p
$ (where $p$ might not occur) and all subterms of $\kappa$ are subconstructors of $[\kappa ]^p$.

\item $[t, [\kappa_1 /\dots/\kappa _n] , s]^p
$ (where $p$ might not occur) and all subterms of $t, \kappa_1 , \dots ,\kappa _n$ and $ s$ are subconstructors of $[t, [\kappa_1 /\dots/\kappa _n] , s]^p
$.
\end{itemize}

The term variables occurring in a term $t$ are all the term variables that appear in $t$ but not immediately to the right of a $\lambda$ or $\all $.
\end{definition}

\begin{definition}[Binder scope, bound variables, free variables]\label{def:free-variables}
%The scope of the displayed occurrence of $\all x :A $  (respectively, $\lam x:A$)  in $\all x :A . \vphi $ (respectively, $\lam x:A .\vphi$) is $\vphi $.
The {\it scope} of the displayed occurrence of $\all x :A $ in $\all x :A . \vphi $  is $\vphi $.
The  {\it scope} of the displayed occurrence of $\lam x:A$ in $\lam x:A .\vphi$ is $\vphi $. 
The  {\it scope} of the displayed occurrence of $\lam x :A $  in $\lam x :A . t $ is $t$.

%An occurrence of the variable $x:A$ in $\vphi $ (respectively, $t$) is bound in $\vphi $ (respectively, $t$) if it occurs inside the scope of an occurrence of $\lam x :A $ or of $\all x :A$ in $\vphi$ (respectively, $t$).

An occurrence of the variable $x:A$ in $\vphi $  is {\it bound in $\vphi $} if it occurs inside the scope of an occurrence of $\lam x :A $ or $\all x :A$ in $\vphi$. An occurrence of the variable $x:A$ in $t$ is {\it bound in $t$} if it occurs inside the scope of an occurrence of $\lam x :A $ or $\all x :A$ in $t$.

A variable occurrence in $\vphi$ is {\it free} if it is not bound in $\vphi$. A variable occurrence in $t$ is {\it free} if it is not bound in $t$.
\end{definition}

\begin{definition}[Substitution]\label{def:substitution}
For any two terms $t$ and $s$ and term variable $x$, the term substitution $t[s/x]$ denotes the term obtained by uniformly and simultaneously replacing all free occurrences of $x$ in $t$ by an occurrence of $s$. 

For any type constructor $\vphi$, term $t$, and term variable $x$, the type constructor substitution $\vphi[t/x]$ denotes the type constructor obtained by uniformly and simultaneously replacing all free occurrences of $x$ in $\vphi$ by an occurrence of $t$. 

For any two type constructors $\vphi$ and $\psi$, and type constructor variable $\alpha$, the type constructor substitution $\vphi[\psi /\alpha ]$ denotes the type constructor obtained by uniformly and simultaneously replacing all occurrences of $\alpha$ in $\vphi$ by an occurrence of $\psi$. 

\end{definition}

Notice that---since a type constructor variable $\alpha$ is never bound by an occurrence of $\Pi, \forall$ or $\lam$---type constructor substitutions, such as  $[\psi/\alpha]$, are neither employed in deductions nor while normalising type constructors. We only introduce these substitutions as technical devices to define type constructor contexts.

\begin{definition}[Term contexts and type constructor contexts]\label{def:contexts}
A term context $\calc[\;]_1\dots [\;]_n$ is a term in which, for each $i\in \{1 , \dots , n\}$, a designated term variable $x_i$ occurs exactly once. By $\calc[\;]_1\dots[t]_m\dots [\;]_n$ with $1\leq m \leq n$ we denote the term obtained by applying the following substitution: $(\calc[\;]_1\dots [\;]_n)[t/x_m]$.

A type constructor context $\calt[\;]_1\dots [\;]_n$ is a type constructor in which, for each $i\in \{1 , \dots , n\}$, a designated  type constructor variable $\alpha_i$ occurs exactly once. By $\calt[\;]_1\dots[\vphi ]_m\dots [\;]_n$ with $1\leq m \leq n$ we denote the type constructor obtained by applying the following substitution: $(\calc[\;]_1\dots [\;]_n)[\vphi  /\alpha_m]$.\end{definition}

\section{The calculus}

{\it Notation}. We use $\Gamma, \Delta $ to denote lists of type and kind assignments for term and type constructor variables, respectively. We use $s,t,u,v$ as metavariables for terms, $A,B,C,D,E,F,T$ for types, $\vphi , \psi, \xi$ for type constructors, and $\Phi, \Psi, \Xi$ for kinds. We use $x,y,z$ for term variables and $\alpha, \beta, \gamma$ for constructor variables. We use $\kappa $ to denote a list of terms. We add subscripts when necessary.

\begin{definition}[Oracle constants and oracular functions]\label{def:oracles}

%{\latent   
%Each oracle constant $o$ has, by definition, either type $A_1\chot \dots\chot A_n$ or type $\all x:C. (A_1\chot \dots\chot A_n)$.
%}

%NEW: 
Each oracle constant $o$ has either arity $0$ or $1$, and is associated to a type $T$ and to an oracular function $f_o(\;,\;)$. 

If $o_0$ is a $0$-ary oracle constant, then $T=\Sigma A$ for any type $A$ and $f_{o_0}(\;,\;)$ is any function that accepts as first input a term context $\calc[\;]_1\dots [\;]_n $ with any number $n$ of holes and as second input a non-zero number $m\leq n$. The output of the function is a closed term of type $A$ that does not contain any oracle constant.

If $o_1$ is a $1$-ary oracle constant, then  $T=\all x:A.\Sigma B$ and $f_{o_1}(\;,\;)$ is any function that accepts as first input a term context $\calc[\;]_1\dots [\;]_n $ with any number $n$ of holes and as second input a non-zero number $m\leq n$. The output of the function is a closed term of type $B[t/x]$ that does not contain any oracle constant, where $t$ is the term to which the $m$th hole of $\calc[\;]_1\dots [\;]_n$ is applied.

\end{definition}
Oracular functions are used to determine the reducti of oracle constants. In particular, as we will see in Table \ref{tab:evaluation}, if the oracular function application $f_o(\calc[\;]_1\dots [\;]_n , m)$ has as value the term $u$, then we use $u$ as reductum of $o$---possibly applied to some argument---when it occurs in the $m$th hole of the context $\calc[\;]_1\dots [\;]_n$. Hence, for instance, if $o$ does not accept arguments, then $\calc[o]_1\dots [o]_n$ will reduce to a term of the form $\calc[\dots]_1\dots[u]_m\dots [\dots ]_n$ if $f_o(\calc[\;]_1\dots [\;]_n , m)=u$. Since the values of $f_o$ depend on a context representing the whole term in which $o$ occurs, there is no need to specify whether $o$ is applied to some argument or not, this information is already contained in the first argument of the function. So, in case $o$ does accept an argument, then the term $\calc[os_1]_1\dots [os_n]_n = \cald[o]_1\dots [o]_n$ will reduce to a term of the form $\calc[\dots ]_1\dots[u]_m\dots [\dots]_n$ if $f_o(\cald[\;]_1\dots [\;]_n , m)=u$. Thus, in general, we use applications $f_o(\calc[\;]_1\dots [\;]_n , m)$ of oracular functions in which $1\leq m\leq n$.

Oracular functions are defined with respect to contexts with possibly several holes because thus we can encode, by using oracle constants, the behaviour of functions implementing learning mechanisms. Indeed, when we reduce a term of the form $\calc[os_1]_1\dots [os_n]_n = \cald[o]_1\dots [o]_n$  to a term of the form $\calc[u_1]_1\dots [u_n]_n$ since $f_o(\cald[\;]_1\dots [\;]_n , i)=u_i$ for $1\leq i\leq n$, we have that the values produced by the different terms $os_i$ can be determined by $f_o$ as if they were produced through a learning process.

\begin{figure*}[t]\centering

\[\vcenter{\infer{\Gamma _1 , \alpha :\Phi  \vdash \alpha  :\Phi}{\Gamma _1 \vdash \Phi :\square }}\qquad \vcenter{\infer{\Gamma _2 , x :A  \vdash x :A}{\Gamma _2 \vdash A :* }}\]where $\alpha$ does not occur in $\Gamma  _1 $ and $x$ does not occur in $\Gamma _2$

\[ \vdash *:\square\qquad \vcenter{\infer{\Gamma \vdash ( \Pi x :A .\Phi  ):\square}{\Gamma , x:A\vdash \Phi  :\square }}\qquad \vcenter{\infer{\Gamma \vdash (\all x:A.B):*}{\Gamma , x:A\vdash B:* }}\]

\[ \infer{\Gamma \vdash \chot A:*}{\Gamma \vdash A:*}\qquad \infer{\Gamma \vdash \Sigma A:*}{\Gamma \vdash A:*}\qquad \infer{\Gamma \vdash A\wedge B:*}{\Gamma \vdash A:*& \Gamma \vdash B:*}\]

\caption{Formation and axiom rules for kinds and types}
\label{tab:formation-axiom-rules}
\end{figure*}

\begin{figure*}[t]\centering

\[ \vcenter{\infer{\Gamma \vdash (\lambda x:{A}. \vphi  ) : (\Pi x:A. \Phi)}{\Gamma , x:A\vdash \vphi  :\Phi }}  \qquad \vcenter{\infer{\Gamma ,   \vdash (\vphi  t):\Phi [t/x]}{\Gamma \vdash  \vphi   :(\Pi x:{A}.\Phi) & \Gamma \vdash t  :A }} \]

\caption{Kind assignment rules}
\label{tab:kind-rules}
\end{figure*}

\begin{figure*}[t]\centering

\[\infer{\Gamma , \alpha :\Psi \vdash \Phi :\square}{\Gamma  \vdash \Psi:\square   & \Gamma \vdash\Phi :\square }\qquad\infer{\Delta , x :A \vdash \Phi :\square}{\Delta  \vdash A:*   & \Delta \vdash\Phi :\square }\]
\[\infer{\Gamma , \alpha :\Psi \vdash \vphi  :\Phi}{\Gamma  \vdash \Psi:\square   & \Gamma \vdash \vphi  :\Phi }\qquad\infer{\Delta , x :A \vdash  \vphi  :\Phi}{\Delta  \vdash A:*   & \Delta \vdash \vphi  :\Phi }\]
\[\infer{\Gamma , \alpha :\Psi \vdash t :A}{\Gamma  \vdash \Psi:\square   & \Gamma \vdash t:A }\qquad\infer{\Delta , x :A \vdash  t:B}{\Delta  \vdash A:*   & \Delta \vdash t:B }\]

where $\alpha$ does not occur free in any element of $\Gamma$ 

and $
x$  does not occur free in any element of  $\Delta$
\caption{Weakening rules}
\label{tab:weakening-rules}
\end{figure*}

\begin{figure*}[t]\centering

\[\infer{\Gamma \vdash (\lambda x:A.t):(\all x:A .B)}{\Gamma , x:A \vdash t:B}\qquad \infer{\Gamma \vdash ts:B[s/x]}{\Gamma  \vdash t : (\all x:A . B) & \Gamma \vdash s:A}\]

\[\vcenter{\infer{\Gamma \vdash \efq(t):P}{\Gamma \vdash t : \bot & \Gamma \vdash P:*}}\quad  \text{where $P$ does not contain $\all$}\]

\[\vcenter{\infer{\Gamma \vdash \langle s,t\rangle :A\wedge B}{\Gamma \vdash t : A & \Gamma \vdash s:B}}\qquad\vcenter{\infer{\Gamma \vdash t\pi_i:A_i}{\Gamma \vdash t:A_0\wedge A_1}}\quad \text{where $i\in \{0,1\}$}\]

%%OLD
%{\latent  \[\vcenter{\infer{\Gamma \vdash t\cho{p}s:A\chot B}{\Gamma  \vdash t:A & \Gamma \vdash s:B }}\qquad  \vdash o_1:\all x :A. (A_1\chot \dots\chot A_n)  \qquad  \vdash o_2: A_1\chot \dots\chot A_n  \]where $o_1$ is an oracle constant of type $\all x :A. (A_1\chot \dots\chot A_n)$ 
%
%and $o_2$ is an oracle constant of type $A_1\chot \dots\chot A_n$}

%%NEW: 
\[ \infer{\Gamma \vdash t\cho{p}s: \chot A}{\Gamma  \vdash t:A & \Gamma \vdash s: A }\qquad  \infer{\Gamma \vdash o:\Sigma A}{\Gamma \vdash A:*}  \]where  $o$ is an oracle constant  and $A$ is the type associated to it by definition

\[ \infer{\Gamma \vdash t\nu : A}{\Gamma \vdash t:\chot A}\qquad  \infer{\Gamma \vdash o\nu :A}{\Gamma \vdash o:\Sigma A}  \]

%%OLD
%{\latent 
%\[\infer{\Gamma \vdash \langle o_1t, \ldots , o_1t \ran : (A_1\chot \dots\chot A_n)^n }{\Gamma \vdash o_1:\all x :A. (A_1\chot \dots\chot A_n) & \Gamma \vdash t :A}\qquad \infer{\Gamma \vdash\langle o_2, \ldots , o_2 \ran : (A_1\chot \dots\chot A_n)^n }{\Gamma \vdash o_2:A_1\chot \dots\chot A_n} \]}

%%NEW: 
\[\infer{\Gamma \vdash o_1t:\Sigma (B[t/x])}{\Gamma \vdash o_1:\all x:A.\Sigma B &\Gamma \vdash t:A}\qquad \infer{\Gamma \vdash\langle o_0\nu, \ldots , o_0\nu \ran : A^n }{\Gamma \vdash o_0\nu:A}\]\[ \infer{\Gamma \vdash \langle (o_1t_1)\nu, \ldots , (o_1t_n)\nu \ran : B[t/x]^n }{\Gamma \vdash o_1:\all x :A. \Sigma B &  \Gamma \vdash t_1:A &\dots & \Gamma \vdash t_n:A} \]where $o_1$ is a $1$-ary oracle constant, $o_0$ is a $0$-ary oracle constant

%{\latent $\langle \dots \rangle $ denotes a series of $n-1$ pairs nested on the left, 
%
%and $A^n$ denotes a conjunction $A\wedge \dots \wedge A$ in which  $A$ occurs $n$ times}

%{\latent 
%\[\vcenter{\infer{\Gamma \vdash ts:B_1[s/x_1]\chot B_2[s/x_2]}{\Gamma  \vdash t:(\all x_1:A_1. B_1)\chot (\all x_2:A_2.B_2) & \Gamma \vdash s:A }}\]where $A$ is a subtype of both $A_1$ and $A_2$
%
%CAN'T WE JUST SIMULATE THIS BY FEEDING THE ELEMENTS OF THE ND-CHOICE BY AN EXTERNAL FORALL? I THINK WE CAN}

%\[\infer{\Gamma \vdash \test_n t: t\Mapsto _{?/n}s}{\Gamma \vdash t:A}\]

%\[\vcenter{\infer{\Gamma \vdash t\pi^*_i:A}{\Gamma \vdash t:(A)^n}} \quad \text{with $1\leq i\leq n$}\]
%{\latent where $\pi^*_i$ denotes the sequence of projections that extract the $i$th element embedded in a series of $n-1$ pairs nested on the left,
%
%and $A^n$ denotes a conjunction $A\wedge \dots \wedge A$ in which  $A$ occurs $n$ times}

%\caption{Type assignment rules}
%\label{tab:type-rules}
%\end{figure*}
%
%\begin{figure*}[t]\centering

\[
\infer[]{\Gamma \vdash \vphi  : \Psi}{\Gamma \vdash \vphi  :\Phi  & \Gamma \vdash \Psi :\square }
\qquad
\infer{\Gamma \vdash s : D}{\Gamma \vdash t :C  & \Gamma \vdash D :* } 
% $\Phi$ is conversion-equivalent to $\Psi$ and 
\]where $\Phi\equiv_{\beta}\Psi$, $C\equiv_{\beta}D$ and $t\equiv_{\beta}s$  according to Definition \ref{def:conv-equiv}

%,  
%
%$\Psi$ is is obtained by replacing any number of occurrences of $C$
%
%in the types of the bound variables of $\Phi$ by $D$, 
%
%and, for any two types $A,B$ such that $A\equiv_{\beta}B$, 
%
%$t '$ is obtained by replacing any number of occurrences of $A$
%
%in the types of the bound variables of $t$ by $B$

%\caption{Conversion rules}
%\label{tab:conversion-rules}
%\end{figure*}

\caption{Type assignment and conversion rules}
\label{tab:type-conv-rules}
\end{figure*}

\begin{definition}[Conversion equivalence $\equiv_{\beta}$]
\label{def:conv-equiv}
Let us define the conversion relation based on $\beta$-reduction for type constructors as follows: 
\[\dots (\lambda x :\vphi  _1 . \vphi _2)t \dots \mapsto \dots  \vphi _2[t/x] \dots \]
where $\dots (\lambda x :\vphi  _1 . \vphi _2)t \dots$ is any type constructor, term or kind.
%
%and $\calk[(\lambda x :\vphi  _1 . \vphi _2)t]\mapsto \calk[\vphi _2[t/x]]$ for any kind $\calk[(\lambda x :\vphi  _1 . \vphi _2)t]$. 
%
Let us, moreover, define $\mapsto ^* $ as the reflexive and transitive closure of $\mapsto$. For any two type constructors, terms or kinds $X$ and $X '$, $X  \equiv_{\beta}X'$ if, and only if, either $X  \mapsto ^* X  '$ or $X  ' \mapsto ^* X  $. 

%For any two type constructors $\vphi  $ and $\vphi  '$, $\vphi  \equiv_{\beta}\vphi  '$ if, and only if, either $\vphi  \mapsto ^* \vphi  '$ or $\vphi  ' \mapsto ^* \vphi  $. 

\end{definition}

Since types in our system are constructed by also employing terms, a simple way to guarantee that a type is well-formed is to build it by applying rules similar to those used to construct terms. The rules in Tables \ref{tab:formation-axiom-rules}, \ref{tab:kind-rules} and \ref{tab:weakening-rules} precisely enable us to do so. In particular, the first three rules in Table \ref{tab:formation-axiom-rules} make it possible, starting from the atomic kind $*$, to construct complex kinds and, starting from an atomic type, to construct complex types by introducing the universal quantifier $\Pi$ for kinds and $\all$ for types. Once we have thus constructed a kind $\Phi$, we can derive an axiom-shaped sequent in which a type constructor variable $\alpha$ of kind $\Phi$ is added to the context and the same type constructor variable $\alpha:\Phi$ occurs on the right-hand side of the sequent. 
Similarly, for types, when we have formed a type $A$, we can derive the axiom-shaped sequent with the variable term  $x:A$ both in the context and on the right-hand side of the sequent.

By using these rules we can obtain complex kinds and types from atomic ones, but we can only assign them to atomic objects: type constructor variables such as $\alpha $ and term variables such as $x$. The kind assignment rules, in Table \ref{tab:kind-rules}, precisely enable us to form complex type constructors while we assign them a suitable kind. The type assignment rules, in Table \ref{tab:type-conv-rules}, enable us to do the same thing but for terms: we can form complex terms by starting from term variables, such as $x$, while we assign them a suitable type. In simply typed $\lambda$-calculus, we only have type assignment rules. Indeed, types in simply typed $\lambda$-calculus can be defined without reference to terms and hence there is no need for a mechanism to handle by levels---in a way suitable for a type theory---the definition of types and terms. It is, indeed,  the reciprocal dependence of terms from types and of types from terms that makes it useful to handle also the construction of types, and not only of terms, by derivation rules. The presented system, in order to orderly stratify the interdependent sets of types and terms, uses the following idea: types that do not depend on terms can be directly obtained, terms can be directly obtained possibly using previously obtained types---those of the variables in the context, obtained by type formation and weakening rules---and, finally, types that do depend on terms can be obtained by applying type constructors to previously obtained terms. In order to construct a type constructor we must first construct the atomic kinds of its type constructor variables, and then construct the type constructor itself by kind assignment rules, Since kinds, just as types in simply typed $\lambda$-calculus, do not depend on terms, they can be directly obtained by kind formation rules and kind assignment rules without resorting to kind constructors.

The fact that we are able to derive type constructors, which are not terms, makes it necessary to have kinds to restrict their behaviour. We cannot use types because otherwise we would mix up terms and type constructors. The presence of kinds makes it, obviously, necessary to also have rules for constructing kinds.

The weakening rules, in Table \ref{tab:weakening-rules}, simply enable us to build up the context from which we construct a kind---first line of the table---derive the kind assignment of a type constructor---second line---and derive the type assignment of a term---third line.   

The conversion rules, Table \ref{tab:type-conv-rules}, enable us to evaluate type constructor applications during the typing phase. This is required since type constructors are meant, not surprisingly, to produce types that can then be assigned to terms. At the same time, type constructors produce types by being applied to arguments and by being evaluated, not unlike terms produce outputs. This generates an apparent contradiction: we need the type constructor to provide us with a type to be used while typing, but the type constructor can only this type during its evaluation, which happens after the typing phase is completed. The adopted solution consists in enabling us to evaluate type constructor applications during the typing phase by the evaluation rules. While we deductively construct our term, we can also construct its type by forming the required type constructors, by applying them and by evaluating them, and we can do all this inside the derivation through which we are constructing our term. Thus we do not need to wait for the dynamic phase of evaluation in order to obtain the type we need to assign to our term.

%{\ehi TRUST INTRO.} 
%We will exploit a similar idea in order to define our trust predicate. Indeed, trust is a relation that essentially depends on our experience of the past behaviour of a program during evaluation. The more often we experience that the program meets our expectations with respect to the frequency of its outputs, the more we will be inclined to trust the program. Hence, the definition of trust must refer to the evaluation---or better, a number of past evaluations---of the program. Since we still want to be able to apply the trust predicate during the typing phase, and not to wait for the end of the dynamic evaluation phase, we also include in the system evaluation rules corresponding to certain reductions of terms. Thus we will be able to study during the typing phase the runtime behaviour of terms and establish statically whether or not we trust the terms. 
%
%Notice that checking whether an application of a rule for $\downarrow $ is correct is easier than to check whether an application of the rule for the evaluation of type constructors is correct. Indeed, for the first kind of rules, it is enough to check whether a term reduces in one step into another, while for the second rule, in general, we need to check whether a type constructor reduces to another one, possibly in several reduction steps.
%

Let us now define, in the standard way, the arrow type $\impl$ as a special case of the product type $\all$.

{\it Notation}. If the variable $x$ does not occur in $B$, then we write $\all x:A. B$ as $A\impl B$. We then define $\neg A $ as $A\impl \bot $ and $A_1\et \dots \et  A_n  \impl C$ as $A_1\impl  \dots \impl   A_n  \impl C$. We, moreover, abbreviate by $A^n$ a conjunction $A\wedge \dots \wedge A$ in which  $A$ occurs $n$ times. Finally, by $\langle t_1 , \dots , t_n  \rangle $ we denote the term $\langle t_1, \langle  t_2, \langle t_3 ,  \dots ,  \langle t_
{n-1}, t_n  \rangle\dots \rangle\rangle \rangle$ containing $n-1$ pairs nested on the right, and by $\pi^*_i$ for $i\in \boldn$ the sequence of projections $\pi_0,\pi_1$ that selects the $i$th element of $\langle t_1 , \dots , t_n  \rangle $.

%and  $\exists x . A$ as $\neg \all x. \neg A$.

\begin{definition}[Derivations, kinding and typing]
Derivations are inductively defined as follows:
\begin{itemize}
\item $\vdash *:\square$ is a derivation with conclusion $\vdash *:\square$ 
\item if $\delta_1 , \dots , \delta _n $ are derivations with conclusions $S_1 , \dots , S _n $ and a rule application of the form\[\infer{S}{S_1 & \dots & S _n}\] is an instance of one of the rules in Tables \ref{tab:formation-axiom-rules}, \ref{tab:kind-rules}, \ref{tab:weakening-rules}, \ref{tab:type-conv-rules}, 
%\ref{tab:eval-pred-rules}, 
then
\[\infer{S}{\delta _1 & \dots & \delta  _n}\]is a derivation with conclusion $S$.
\end{itemize}

For any kindless type constructor $\vphi$ and kind $\Phi$, if there exists a derivation with conclusion $\Gamma \vdash \vphi :\Phi$, we say that the type constructor $\vphi $ has kind $\Phi$ and we express it formally by $\vphi : \Phi$.

% such that there is no type constructor $\vphi ' $ such that $\vphi \mapsto \vphi '$. 

If $\vphi:*$, then $\vphi$ is a type.

For any typeless term $t$ and type $A :*$, if there exists a derivation with conclusion $\Gamma \vdash t:A$, we say that the term $t$ has type $A$ and we express it formally by $t:A$.
\end{definition}

\section{The evaluation relation and its basic properties}
\label{sec:reductions}

By Definition \ref{def:conv-equiv}, while defining typing derivations, we already defined the evaluation relation for type constructors. Indeed, we need to be able to statically evaluate type constructors if we wish to construct terms and assign them a type which is already in normal form. Nevertheless, we still need to define the evaluation relation that enables us to formalise the actual dynamic process of execution of a term. This relation, unlike the one between type constructors, formalises a probabilistic notion of computation and, therefore, is a ternary relation $t\mapsto ^p t'$ intuitively indicating that the term $t$ evaluates to the term $t'$ with a probability of $p$, where $p\in \boldq$ and $0\leq p\leq 1$.

The reduction rules for evaluating terms are shown in Table \ref{tab:evaluation}. Just as evaluation rules for type constructors, These rules can be applied at any depth inside a term.\footnote{Notice that we use the words {\it reduction} and {\it evaluation} interchangeably.} That is, if a rule of the form $t\mapsto_p t'$ is displayed in Table \ref{tab:evaluation}, then, for any term $s$ that contains at least one occurrence of $x$, $s[t/x] \mapsto _p s[t'/x]$.

\begin{figure*}[t]\centering
\[(\lambda x:A.t)s\mapsto_1 t[s/x]\]
\[\langle t_0,t_1\rangle\pi_i\mapsto_1 t_i \quad \text{for $i\in \{0,1\}$}\]
\[(t\cho{p}s)\nu \mapsto_p t \qquad  (t\cho{p}s)\nu \mapsto_{1-p} s\]

\[\calc[o_0\nu]_1 \dots [o_0\nu]_n\mapsto _{1}\calc[s_1]_1 \dots [s_n]_n \]where $o_0$ is a $0$-ary oracle constant,

and $f_{o_0}(\calc[\;]_1\dots[\;]_n, i )=s_i$ 
for $1\leq i\leq n$

\[\calc[(o_1t_1)\nu]_1\dots[(o_1t_n)\nu]_n \mapsto _{1}\calc[s_1]_1\dots[s_n]_n \]where $o_1$ is a $1$-ary oracle constant,

$\calc[(o_1t_1)\nu]_1\dots[(o_1t_n)\nu]_n=
\cald[o_1]_1\dots[o_1]_n$,

and $f_{o_1}(\cald[\;]_1\dots[\;]_n, i )=s_i$ for $1\leq i\leq n$

%\[ \test_n (ot): ot\Mapsto _{?/n}s\mapsto_p \lan s_1, \dots , s_n \ran : ot\Mapsto _{m/n}s \]where $p$ is the probability of obtaining $s_1, \dots , s_n$ from $n$ executions of the oracle constant $o$ (computed from the definition of $o$) and  $m$ is the number of occurrences of $s$ in $s_1 \dots , s_n$

%{\latent \[[t, \dots ,s ]^pt \quad\mapsto_p\quad [t, \dots ,s]^p \] USELESS IF WE REMOVE THE ELIMINATION RULE FOR $\Mapsto$}
%%\[ [[\kappa _1  /\dots /t,\kappa _i /\dots /\kappa_n ], s ]^pt : u \quad\mapsto_p\quad [t,\kappa _i, s ]^p: u\]

\caption{Evaluation reductions}
\label{tab:evaluation}
\end{figure*}

We now show that $\mapsto$ preserves kinds and that $\mapsto^p$ preserves typability.
\begin{theorem}[Subject reduction]\label{thm:subject-reduction}
For any type constructor $\vphi :\Phi$, if $\vphi\mapsto \vphi'$, then $\vphi':\Phi$ and all free term variables of $\vphi '$ occur free also in $\vphi$.

For any term $t : A$, if $t\mapsto_p t'$, then $t':A'$---where $A'$ has the same propositional and quantificational structure as $A$ but might contain different terms---and all free term variables of $t '$ occur free also in $t$.
\end{theorem}
\begin{proof}
%NOW FALSE: Before beginning to prove the statements, we remind the reader that, as stated in Definition \ref{def:subconstructor}, the term variables {\it occurring} in a  type constructor $\vphi$ are only those variables $x_1 , \dots , x_n $ such that there exists a subconstructor of the form $\alpha x_1 \dots x_n $ which occurs inside $\vphi$---e.g., $x$ does not occur in $\lambda y:(\alpha  t).\lambda x:\beta .\gamma y $, only $y$ does, and it does not occur in $\lambda y:(\alpha  x).\gamma y $ either.
%Moreover, as stated in Definition \ref{def:subterm}, the variables {\it occurring} in a term  are only those variables that are subterms of the term---e.g., $x$ does not occur in $\lambda y:(\beta t).\lambda x:\alpha .y $, only $y$ does, and it does not occur in $\lambda y:(\beta x).y $ either.

We start by proving the part of the statement concerning type constructors. The only possible evaluation rule applicable to a type constructor is the following:\[\calt [(\lambda x: \vphi _1.\vphi _2)t]\mapsto \calt[\vphi _2[t/x]]\]

Now, suppose that $\vphi _2:\Phi_2$. Then the kind of $\lambda x: \vphi _1.\vphi _2$ must be $\Pi x: \vphi _1.\Phi _2$. This is clear by simple inspection of the rules in Table \ref{tab:kind-rules} for assigning kinds. As a consequence, for any $t:\vphi_1$, the kind of $(\lambda x: \vphi _1.\vphi _2)t$ is  $\Phi _2$. The kind of $\vphi _2[t/x]$ is  $\Phi _2$ too since the substitution replaces $x:\vphi _1$ by $t:\vphi_1$ inside $\vphi _2 : \Phi _2$. Therefore, $(\lambda x: \vphi _1.\vphi _2)t$ and  $\vphi _2[t/x]$ always have the same kind and---since replacing a subconstructor with a kind by a subconstructor with the same kind inside a constructor leaves the kind of the constructor unchanged---also the constructors $\calt [(\lambda x: \vphi _1.\vphi _2)t] $ and $  \calt[\vphi _2[t/x]]$ have the same kind.

This reduction, moreover, does not free any variable:
\begin{itemize}
\item All variables bound in $\vphi _2$ before the reduction, remain bound after it. 
\item All variables bound by occurrences of $\lambda$ occurring in $\calt[\;]$ before the reduction, remain bound after it.
\item All occurrences of $x$ in $\vphi _2$ are replaced by $t$ during the reduction.
\end{itemize}

We prove now the part of the statement concerning terms. We reason by cases on the reduction rule that has the reduction $t\mapsto _pt'$ as an instance.

\begin{itemize}
\item $\calc[(\lambda x:A.s)t]\mapsto_1 \calc[s[t/x]]$. 
%SHORT VERSION: This case is handled in the same way as the one of the reduction rule for type constructors, with the obvious changes in metavariable notation, $\all x $ instead of $\Pi x $, and the additional remark that the free variables of $t_1$ occurring in $t_2[t_1/x]$ occur also in $(\lambda x:A.t_2)t_1$. 

Suppose that $s:B$. Then the type of $\lambda x:A.s$ must be $\all x:A.B$. This is clear by simple inspection of the rules in Table \ref{tab:type-conv-rules} for assigning types. As a consequence, for any $t:A$, the type of $(\lambda x:A.s)t$ is  $B$. The type of $s[t/x]$ is $B$ too since the substitution replaces $x:A$ by $t:A$ inside $s:B$. Therefore, $(\lambda x: A.s)t$ and  $s[t/x]$ always have the same type and---since replacing a subterm with a type by a subterm with the same type inside a term leaves the type of the term unchanged---also the terms $\calc [(\lambda x:A.s)t] $ and $ \calc[s[t/x]]$ have the same type.

This reduction does not free any variable:
\begin{itemize}
\item All variables bound in $s$ before the reduction, remain bound after it.
\item All variables bound by occurrences of $\lambda$ occurring in $\calc[\;]$ before the reduction, remain bound after it.
\item All occurrences of $x$ in $s$ are replaced by $t$ during the reduction.
\end{itemize}

\item $\calc[\langle t_0,t_1\rangle\pi_i]\mapsto_1 \calc[t_i]$. 

Suppose that $t_0:A_0$ and $t_1:A_1$. Then the type of $\langle t_0,t_1\rangle$ must be $A_1 \et A_2$. This is clear by simple inspection of the rules in Table \ref{tab:type-conv-rules} for assigning types. As a consequence, the type of $\langle t_0,t_1\rangle\pi_i$ is the type of $t_i$, that is, $A_i$. Since replacing a subterm with a type by a subterm with the same type inside a term leaves the type of the term unchanged, also the terms $\calc [(\lambda x:A.s)t] $ and $ \calc[s[t/x]]$ have the same type.

This reduction does not free any variable:
\begin{itemize}
\item All variables bound in $t_i$ before the reduction, remain bound after it.
\item All variables bound by occurrences of $\lambda$ occurring in $\calc[\;]$ before the reduction, remain bound after it.\end{itemize}

\item $\calc[(t_0\cho{p}t_1)\nu] \mapsto_q \calc[t_i]$.

Suppose that $t_0:A_0$ and $t_1:A_1$. Then the type of $t_0\cho{p}t_1$ must be $\chot A$ where $A=A_0=A_1$. This is clear by simple inspection of the rules in Table \ref{tab:type-conv-rules} for assigning types. As a consequence, the type of $(t_0\cho{p}t_1)\nu$ is the type of $t_i$, that is, $A_i=A$. Since replacing a subterm with a type by a subterm with the same type inside a term leaves the type of the term unchanged, also the terms $\calc[(t_0\cho{p}t_1)\nu] $ and $ \calc[t_i]$ have the same type.

This reduction does not free any variable:
\begin{itemize}
\item All variables bound in $t_i$ before the reduction, remain bound after it.
\item All variables bound by occurrences of $\lambda$ occurring in $\calc[\;]$ before the reduction, remain bound after it.\end{itemize}

\item $\calc[\calc'[o_0\nu]_1 \dots [o_0\nu]_n]\mapsto _{1}\calc[\calc'[s_1]_1 \dots [s_n]_n] $.

Suppose that the types of $s_1, \dots , s_n$ are $A_1, \dots , A_n$. Then we must have that $A_1= \dots =A_n$ and that the type of     $o_0\nu$ is $A=A_1= \dots =A_n$. The type of $o_0$ must then be $\Sigma A$. This is clear by Definition \ref{def:oracles} and by inspection of the rules in Table \ref{tab:type-conv-rules} for assigning types. As a consequence, the type of $o_0\nu$ and the type of each $s_1, \dots , s_n$ is the same, that is, $A=A_1= \dots =A_n$. Since replacing a subterm with a type by a subterm with the same type inside a term leaves the type of the term unchanged, also the terms $\calc[\calc'[o_0\nu]_1 \dots [o_0\nu]_n]$ and $\calc[\calc'[s_1]_1 \dots [s_n]_n]$ have the same type.

This reduction does not free any variable:
\begin{itemize}
\item All variables occurring in $s_1, \dots , s_n$ are bound according to Definition \ref{def:oracles}.
\item All variables bound by occurrences of $\lambda$ occurring in $\calc[\calc'[\;]_1 \dots [\;]_n]$ before the reduction, remain bound after it.\end{itemize}

\item $ \calc[\calc'[(o_1t_1)\nu]_1\dots[(o_1t_n)\nu]_n] \mapsto _{1}\calc[\calc'[s_1]_1\dots[s_n]_n] $. 

Suppose that the types of $s_1, \dots , s_n$ are $B_1, \dots , B_n$. Then we must have that $B_1= \dots =B_n$ and that the type of     $(o_1t_1)\nu, \dots , (o_1t_n)\nu$ is $B=B_1= \dots =B_n$. The type of $o_1$ must then be of the form $\all x:A.\Sigma B$ and we must have that $t_1:A, \dots , t_n:A$. This is clear by Definition \ref{def:oracles} and by inspection of the rules in Table \ref{tab:type-conv-rules} for assigning types. As a consequence, the type of each $(o_1t_1)\nu , \dots , (o_1t_n)\nu$ and the type of each $s_1, \dots , s_n$ is the same, that is, $B=B_1= \dots =B_n$. Since replacing a subterm with a type by a subterm with the same type inside a term leaves the type of the term unchanged, also the terms $\calc[\calc'[o_0\nu]_1 \dots [o_0\nu]_n]$ and $\calc[\calc'[s_1]_1 \dots [s_n]_n]$ have the same type.

This reduction does not free any variable:
\begin{itemize}
\item All variables occurring in $s_1, \dots , s_n$ are bound according to Definition \ref{def:oracles}.
\item All variables bound by occurrences of $\lambda$ occurring in $\calc[\calc'[\;]_1 \dots [\;]_n]$ before the reduction, remain bound after it.\end{itemize}
\end{itemize}
\end{proof}

We obviously cannot prove that a term only has one normal form. Indeed, a term of the form $s\cho{p}t$, for $s\neq t$, is {\it supposed} to have at least two normal forms: the normal form of $s$ and that of $t$. Hence, we will simply prove that type constructors have a unique normal form. In Section \ref{sec:normalisation}, we will show that each type constructor has a normal form that can be reached in a finite number of evaluation steps. We will actually prove much more than that, but the rest will be discussed in due time. Right now, we prove Theorem \ref{thm:confluence}, which guarantees that all reduced forms of a term can be evaluated to the same normal form. First, though, we need to prove a couple of lemmata that enable us to handle, inside the proof of Theorem \ref{thm:confluence}, the substitutions generated by the reduction of applications.

\begin{lemma}\label{lem:substitution}
For any two variables $x$ and $y$ such that $x\neq y$, term $t$ in which $y$ does not occur free and type constructor or term $\theta$, $(\theta [t/x])[(s[t/x])/y]=(\theta [s/y])[t/x]$.
\end{lemma}
\begin{proof}
The proof is by induction on the number of symbol occurrences in $\theta$. We reason by cases on the structure of $\theta$.
%{\ehi DOUBLE CHECK SUBS ON TYPE OF TERM VARIABLE}
\begin{itemize}
\item $\theta = \alpha$. Since $(\alpha [t/x])[(s[t/x])/y]=\alpha [(s[t/x])/y]=\alpha=\alpha[t/x]=(\alpha [s/y])[t/x]$ we have that the statement holds.

\item $\theta = \all z:A. \psi$. Considering that we can always rename bound variables, we suppose that $x\neq z\neq y$. Then{\footnotesize \[((\all z:A. \psi) [t/x])[(s[t/x])/y]= (\all z:(A[t/x]). \psi[t/x]) [(s[t/x])/y] = \all z:((A[t/x])[(s[t/x])/y]). (\psi[t/x]) [(s[t/x])/y]\]}By inductive hypothesis, $(\psi[t/x]) [(s[t/x])/y] = (\psi [s/y])[t/x]$ and $(A[t/x]) [(s[t/x])/y] = (A [s/y])[t/x]$, which implies that $\all z:((A[t/x]) [(s[t/x])/y]). (\psi[t/x]) [(s[t/x])/y]= \all z:((A [s/y])[t/x]). (\psi [s/y])[t/x]$. From this and the following equalities:\[\all z:((A [s/y])[t/x]). (\psi [s/y])[t/x]= (\all z: (A [s/y]). (\psi [s/y]))[t/x]=((\all z:A. \psi )[s/y])[t/x]\]---also due to the fact that $x\neq z\neq y$---we can conclude that the statement holds in this case as well. 

\item $\theta = \chot \psi$. We have that\[((\chot \psi) [t/x])[(s[t/x])/y]= (\chot \psi[t/x]) [(s[t/x])/y] = \chot (\psi[t/x]) [(s[t/x])/y]\]By inductive hypothesis, $(\psi[t/x]) [(s[t/x])/y] = (\psi [s/y])[t/x]$, which implies that $\chot (\psi[t/x]) [(s[t/x])/y]= \chot (\psi [s/y])[t/x]$. From this and the fact that $\chot (\psi [s/y])[t/x]= (\chot (\psi [s/y]))[t/x]=((\chot \psi )[s/y])[t/x]$ we can conclude that that the statement holds in this case as well.

\item $\theta = \psi_1 \wedge \psi_2$. We have that\[(( \psi_1 \wedge \psi_2) [t/x])[(s[t/x])/y]= ( \psi_1[t/x] \wedge \psi_2[t/x]) [(s[t/x])/y] \]\[ =(\psi_1[t/x])[(s[t/x])/y] \wedge (\psi_2[t/x])[(s[t/x])/y]\]By inductive hypothesis, $(\psi_i[t/x]) [(s[t/x])/y] = (\psi_i [s/y])[t/x]$, which implies that\[(\psi_1[t/x])[(s[t/x])/y] \wedge (\psi_2[t/x])[(s[t/x])/y] = (\psi_1 [s/y])[t/x] \wedge (\psi_2 [s/y])[t/x]\]From this and the fact that\[(\psi_1 [s/y])[t/x] \wedge (\psi_2 [s/y])[t/x]=(\psi_1 [s/y] \wedge \psi_2 [s/y])[t/x]= ((\psi_1 \wedge \psi_2)[s/y])[t/x]\]we can conclude that that the statement holds in this case as well.

\item $\theta = \lambda z:A . \psi$. Considering that we can always rename bound variables, we suppose that $x\neq z\neq y$. Then{\footnotesize \[((\lambda z:A. \psi) [t/x])[(s[t/x])/y]= (\lambda z:(A[t/x]). \psi[t/x]) [(s[t/x])/y] = \lambda z:((A[t/x]) [(s[t/x])/y]). (\psi[t/x]) [(s[t/x])/y]\]}By inductive hypothesis, $(\psi[t/x]) [(s[t/x])/y] = (\psi [s/y])[t/x]$ and $(A[t/x]) [(s[t/x])/y] = (A [s/y])[t/x]$, which implies that $\lambda z:((A[t/x]) [(s[t/x])/y]). (\psi[t/x]) [(s[t/x])/y]= (\lambda z: (A[s/y]). (\psi [s/y]))[t/x]$. From this and the following equalities: $\lambda z:((A[s/y])[t/x]). (\psi [s/y])[t/x]= (\lambda z:(A[s/y]). (\psi [s/y]))[t/x]=((\lambda z:A. \psi )[s/y])[t/x]$---also due to the fact that $x\neq z\neq y$---we can conclude that the statement holds in this case as well.

\item $\theta = \psi u$. We have that{\footnotesize \[(( \psi u) [t/x])[(s[t/x])/y]= ( \psi[t/x] u[t/x]) [(s[t/x])/y] = (\psi[t/x])[(s[t/x])/y] (u[t/x])[(s[t/x])/y]\]}By inductive hypothesis, $(\psi[t/x])[(s[t/x])/y]  = (\psi [s/y])[t/x]$ and $(u[t/x])[(s[t/x])/y]=(u [s/y])[t/x]$, which implies that\[(\psi [t/x])[(s[t/x])/y] ( u[t/x])[(s[t/x])/y] = (\psi  [s/y])[t/x] ( u [s/y])[t/x]\]From this and the fact that\[(\psi  [s/y])[t/x]  ( u [s/y])[t/x]=(\psi  [s/y]   u [s/y])[t/x]= ((\psi  u)[s/y])[t/x]\]we can conclude that that the statement holds in this case as well.

\item $\theta =z$. We have three subcases: $x\neq z\neq y$, $z=x$ and $z=y$. 

If $x\neq z\neq y$ then $(z [t/x])[(s[t/x])/y]=z =(z [s/y])[t/x]$. 

If $\theta =z=x$ then $(x [t/x])[(s[t/x])/y]= t [(s[t/x])/y]=t$ since, by assumption, $y$ does not occur free in $t$. But we also have that    $(x [s/y])[t/x]=x [t/x]$ since $x\neq y$ and, finally, $x [t/x]=t$. Hence, indeed, $(x [t/x])[(s[t/x])/y]=(x [s/y])[t/x]$  as desired since $x=\theta$.

If $\theta =z=y$ then $(y [t/x])[(s[t/x])/y]= y [(s[t/x])/y]=s[t/x]$ since, by assumption, $y\neq x$. But we also have that    $(y [s/y])[t/x]=s [t/x]$. Hence, indeed, $(y [t/x])[(s[t/x])/y]=(y [s/y])[t/x]$ as desired since $y=\theta$.

\item $\theta =\lan u_1,u_2\ran$. We have that\[(\lan u_1,u_2\ran [t/x])[(s[t/x])/y]= \lan u_1[t/x] , u_2 [t/x]\ran [(s[t/x])/y] \]\[ =\lan (u_1[t/x])[(s[t/x])/y] , (u_2[t/x])[(s[t/x])/y]\ran \]By inductive hypothesis, $(u_i[t/x]) [(s[t/x])/y] = (u_i [s/y])[t/x]$, which implies that\[\lan (u_1[t/x])[(s[t/x])/y] , (u_2[t/x])[(s[t/x])/y]\ran  = \lan (u_1 [s/y])[t/x] , (u_2 [s/y])[t/x]\ran \]From this and the fact that\[\lan (u_1 [s/y])[t/x] , (u_2 [s/y])[t/x]\ran =\lan u_1 [s/y] , u_2 [s/y]\ran [t/x]= (\lan u_1  , u_2\ran[s/y])[t/x]\]we can conclude that that the statement holds in this case as well.

\item $\theta =u\pi_0$ and $t=u\pi_1$. We have that{\footnotesize \[((  u\pi_i) [t/x])[(s[t/x])/y]= (  (u [t/x]) \pi_i) [(s[t/x])/y] = (( u [t/x])[(s[t/x])/y] )\pi_i \]}By inductive hypothesis, $( u [t/x])[(s[t/x])/y]  = ( u  [s/y])[t/x]$, which implies that\[(( u  [t/x])[(s[t/x])/y]) \pi_i  = (( u   [s/y])[t/x]) \pi_i \]From this and the fact that\[(( u   [s/y])[t/x])  \pi_i =( (u   [s/y]   )\pi_i)[t/x]= ((u  \pi_i)[s/y])[t/x]\]we can conclude that that the statement holds in this case as well.

\item $\theta =\lambda z:A. u$. Considering that we can always rename bound variables, we suppose that $x\neq z\neq y$. Then{\footnotesize \[((\lambda z:A.  u ) [t/x])[(s[t/x])/y]= (\lambda z:(A[t/x]).  u [t/x]) [(s[t/x])/y] = \lambda z:((A[t/x]) [(s[t/x])/y]). ( u [t/x]) [(s[t/x])/y]\]}By inductive hypothesis, $( u [t/x]) [(s[t/x])/y] = ( u  [s/y])[t/x]$ and $(A[t/x]) [(s[t/x])/y] = (A [s/y])[t/x]$, which implies that $\lambda z:((A[t/x]) [(s[t/x])/y]). ( u [t/x]) [(s[t/x])/y]= (\lambda z: (A[s/y]). ( u  [s/y]))[t/x]$. From this and the following equalities: $\lambda z:((A[s/y])[t/x]). ( u  [s/y])[t/x]= (\lambda z:(A[s/y]). ( u  [s/y]))[t/x]=((\lambda z:A.  u  )[s/y])[t/x]$---also due to the fact that $x\neq z\neq y$---we can conclude that the statement holds in this case as well.

\item $\theta =vu$. We have that{\footnotesize \[((  v  u) [t/x])[(s[t/x])/y]= (  v [t/x] u[t/x]) [(s[t/x])/y] = ( v [t/x])[(s[t/x])/y] (u[t/x])[(s[t/x])/y]\]}By inductive hypothesis, $( v [t/x])[(s[t/x])/y]  = ( v  [s/y])[t/x]$ and $(u[t/x])[(s[t/x])/y]=(u [s/y])[t/x]$, which implies that\[( v  [t/x])[(s[t/x])/y] ( u[t/x])[(s[t/x])/y] = ( v   [s/y])[t/x] ( u [s/y])[t/x]\]From this and the fact that\[( v   [s/y])[t/x]  ( u [s/y])[t/x]=( v   [s/y]   u [s/y])[t/x]= ((v  u)[s/y])[t/x]\]we can conclude that that the statement holds in this case as well. 

\item $\theta = u_1\cho{p}u_2$.
We have that\[( (u_1\cho{p}u_2) [t/x])[(s[t/x])/y]= (( u_1[t/x]) \cho{p} (u_2 [t/x])) [(s[t/x])/y] \]\[ = \cho{p} (u_2[t/x])[(s[t/x])/y]\]By inductive hypothesis, $(u_i[t/x]) [(s[t/x])/y] = (u_i [s/y])[t/x]$, which implies that\[ ((u_1[t/x])[(s[t/x])/y] )\cho{p} ((u_2[t/x])[(s[t/x])/y] ) =  ((u_1 [s/y])[t/x] )\cho{p} ((u_2 [s/y])[t/x]) \]From this and the fact that\[ ((u_1 [s/y])[t/x]) \cho{p} ((u_2 [s/y])[t/x]) =( (u_1 [s/y]) \cho{p} (u_2 [s/y])) [t/x]= (( u_1  \cho{p} u_2)[s/y])[t/x]\]we can conclude that that the statement holds in this case as well. 

\item $\theta =u\nu$. We have that{\footnotesize \[((  u\nu) [t/x])[(s[t/x])/y]= (  (u [t/x]) \nu) [(s[t/x])/y] = (( u [t/x])[(s[t/x])/y] )\nu \]}By inductive hypothesis, $( u [t/x])[(s[t/x])/y]  = ( u  [s/y])[t/x]$, which implies that\[(( u  [t/x])[(s[t/x])/y]) \nu  = (( u   [s/y])[t/x]) \nu \]From this and the fact that\[(( u   [s/y])[t/x])  \nu =( (u   [s/y]   )\nu)[t/x]= ((u  \nu)[s/y])[t/x]\]we can conclude that that the statement holds in this case as well.

\end{itemize}

\end{proof}

\begin{lemma}\label{lem:reduction-substitution}
For any type constructor $\chi $ and term $t$ that does not contain free variables which are bound in $\chi$, if $\chi \mapsto \chi '$ then $\chi[t/x]\mapsto \chi'[t/x]$.
\end{lemma}
\begin{proof}
We reason by induction on the number of symbol occurrences in  $\chi$.

In the base case, $\chi=\alpha $. Since $\alpha$ cannot be reduced to any type constructor, the statement is trivially satisfied.

%{\ehi DOUBLE CHECK}

Suppose now that all type constructors with less than $n$ symbols verify the statement, we show that also a type constructor $\chi$ with $n$ symbols does.  We reason by cases on the structure of $\chi$.
\begin{itemize}
\item $\chi = \all y:A . \psi$. First, if $x$ is bound in $\chi$ then $\chi[t/x]=\chi$, $\chi'[t/x]=\chi'$ and the statement is tautologous. We hence suppose that $x\neq y$. Now, $\all y:A .  \psi\mapsto \all y:A .  \psi'$ if, and only if, $ \psi\mapsto  \psi'$ and $\all y:A .  \psi$ cannot be reduced in any other way. By inductive hypothesis, if $\psi \mapsto \psi '$ then $\psi[t/x]\mapsto \psi '[t/x]$. Hence, we have that $\psi[t/x]\mapsto \psi '[t/x]$  enables us to obtain $\all y:(A[t/x]) . (\psi [t/x]) \mapsto \all y:(A[t/x]) . (\psi'[t/x])$. Now, since  $x\neq y$ and since $t$ does not contain free variables which are bound in $\chi$, we have that $\all y:(A[t/x]) . (\psi [t/x])=(\all y:A .  \psi )[t/x]$ and $\all y:(A[t/x]) . (\psi '[t/x])=(\all y:A .  \psi ')[t/x]$ and thus that the statement holds for $\chi = \all y:A . \psi$.

\item $\chi = \chot \psi$. Now, $\chot \psi\mapsto \chot \psi'$ if, and only if, $ \psi\mapsto  \psi'$ and $\chot \psi$ cannot be reduced in any other way. By inductive hypothesis, if $\psi \mapsto \psi '$ then $\psi[t/x]\mapsto \psi '[t/x]$. But since $\psi[t/x]\mapsto \psi '[t/x]$  implies that $\chot(\psi [t/x]) \mapsto \chot(\psi'[t/x])$ and since $\chot(\psi [t/x])=(\chot \psi )[t/x]$ and $\chot(\psi '[t/x])=(\chot \psi ')[t/x]$, we have that the statement holds for $\chi = \chot \psi$.

\item $\chi = \psi_1 \wedge \psi_2$. Now, $ \psi_1 \wedge \psi_2$ can only reduce in one step to terms of the form $ \psi_1  ' \wedge \psi_2 $ or $ \psi_1 \wedge  \psi_2'$   such that $\psi_i\mapsto  \psi_i'$ for $i\in \{1,2\}$. Moreover, $\psi_i\mapsto  \psi_i'$ implies that $ \psi_1 \wedge \psi_2 \mapsto   \psi_1' \wedge \psi_2 $ and that  $ \psi_1 \wedge \psi_2 \mapsto   \psi_1 \wedge \psi_2 '$. Moreover, by inductive hypothesis, if $\psi_i\mapsto  \psi_i'$ then $\psi_i [t/x]\mapsto  \psi_i'[t/x]$. Therefore, if $\chi\mapsto \chi'$, then 
\begin{itemize}
\item either $\chi'=   \psi_1' \wedge \psi_2 $ and $\chi[t/x]= (\psi_1 \wedge \psi_2)[t/x]=\psi_1[t/x] \wedge \psi_2[t/x] \mapsto  \psi_1'[t/x] \wedge \psi_2[t/x]= (\psi_1' \wedge \psi_2)[t/x]=\chi'[t/x]$,

\item or $\chi'=   \psi_1 \wedge \psi_2'$ and $\chi[t/x]= 
 (\psi_1 \wedge \psi_2)[t/x]=\psi_1[t/x] \wedge \psi_2[t/x] \mapsto  \psi_1[t/x] \wedge \psi_2'[t/x]= 
 (\psi_1 \wedge \psi_2')[t/x]=\chi'[t/x]$
\end{itemize}In both cases, the statement holds for $\chi = \psi_1 \wedge \psi_2$.

\item $\chi = \lambda y:A . \psi$. First, if $x$ is bound in $\chi$ then $\chi[t/x]=\chi$, $\chi'[t/x]=\chi'$ and the statement is tautologous. Hence, we suppose that $x$ is free in $\chi$. Now, it is clear that $\lambda y:A . \psi\mapsto \lambda y:A . \psi'$ if, and only if, $ \psi\mapsto  \psi'$ and that $\lambda y:A . \psi$ cannot be reduced in any other way. Therefore we know that $\chi\mapsto \chi $ implies that $ \psi\mapsto  \psi'$ By inductive hypothesis, moreover,  $\psi \mapsto \psi '$ implies that $\psi[t/x]\mapsto \psi '[t/x]$. Since we assumed that $x$ is free in $\chi$ we have, as a consequence, that $\lambda y:(A[t/x]) .(\psi [t/x])=(\lambda y:A . \psi )[t/x]$ and $\lambda y:(A[t/x]) .(\psi '[t/x])=(\lambda y:A . \psi ')[t/x]$. But  then from $\psi[t/x]\mapsto \psi '[t/x]$---again from the fact that $\lambda y:A . \theta\mapsto \lambda y:A . \theta'$ if, and only if, $ \theta\mapsto  \theta'$---we can obtain $(\lambda y:A . \psi )[t/x]=\lambda y:(A[t/x]) .(\psi [t/x]) \mapsto \lambda y:(A[t/x]) .(\psi'[t/x])=(\lambda y:A . \psi ')[t/x]$, which proves that the statement holds in this case as well.

\item $\chi = \psi s$. Now, either $\psi s$ reduces to $\psi' s$---in which case we can prove the statement by induction just like in the first case considered---or we have that $\chi = \psi s=(\lambda y:A.\theta)s\mapsto\theta [s/y] $. In this case we have to prove that $((\lambda y:A.\theta)s)[t/x]\mapsto (\theta [s/y])[t/x] $. As in the previous case, we suppose that $x$ is free in $\chi$ because otherwise $\chi[t/x]=\chi$, $\chi'[t/x]=\chi'$ and the statement is tautologous. But if $x$ is free in $\chi = \psi s=(\lambda y:A.\theta)s  $ then $x\neq y$. Thus, $ ((\lambda y:A.\theta)s)[t/x]=  (\lambda y:A.(\theta [t/x]) )(s[t/x])\mapsto (\theta [t/x])[(s[t/x])/y]$. Since---with the help of the assumption that $y$ does not occur in $t$ and of the fact that $x\neq y$---Lemma \ref{lem:substitution} guarantees us that $ (\theta [t/x])[(s[t/x])/y]=(\theta [s/y])[t/x] $, we have that the statement holds also in this case. 
\end{itemize}
\end{proof}

\begin{theorem}[Confluence for type constructors]\label{thm:confluence}
For any $\vphi:\Phi$, if $\vphi \mapsto^* \vphi_1$ and $\vphi \mapsto^* \vphi_2$, then there exists a type constructor $\vphi_3$ for which $\vphi _1 \mapsto^* \vphi_3$ and $\vphi_2 \mapsto^* \vphi_3$.
\end{theorem}
\begin{proof}
We are going to prove the following statement:

\begin{quote}If $\vphi \mapsto \vphi_1$ and $\vphi \mapsto \vphi_2$, then we can always reduce $\vphi_1$ and $\vphi_2$ in such a way to show that $\vphi _1 \mapsto  \vphi_3$ and $\vphi_2 \mapsto  \vphi_3$.
\end{quote}

Then a simple argument by diagram chase---see, for instance, \cite[Chap. 1, Sec. 1.4]{sor06}---yields the statement.

Let us reason by induction on the number of symbol occurrences in $\vphi$ and by cases on the reductions $\vphi \mapsto \vphi_1$ and $\vphi \mapsto \vphi_2$.

In the base case, $\varphi=\alpha $, which is impossible since $\alpha$ cannot be reduced.

Suppose now that all type constructors with less than $n$ symbols verify the statement, we show that also a type constructor $\vphi$ with $n$ symbols does.  

\begin{itemize}
%\item $\vphi = \Sigma \psi$. By inductive hypothesis, if $\psi \mapsto \psi _1 $ and $\psi \mapsto \psi _2 $, then there exists a $\psi_3$ for which $\psi _1 \mapsto  \psi _3$ and $\psi _2 \mapsto  \psi_3$. Since $\Sigma \psi $ can only reduce to terms of the form $\Sigma \psi'$ such that $\psi\mapsto  \psi'$ and since $\psi\mapsto  \psi'$  implies that $\Sigma \psi \mapsto  \Sigma \psi'$, we have that also $\vphi = \Sigma \psi$ verifies the statement.

\item $\vphi = \all x:A. \psi$. 

Now, clearly, $\all x:A. \psi$ can only reduce to terms of the form $\all x:A. \psi \mapsto \all x:A. \psi'$ such that $\psi\mapsto  \psi'$ and, moreover, $\psi\mapsto  \psi'$  implies that $\all x:A. \psi \mapsto \all x:A. \psi'$. Therefore, any two reductions $\all x:A. \psi=\vphi \mapsto \vphi_1 $ and $\all x:A. \psi=\vphi \mapsto \vphi_2$ can be represented as follows: 
\[\vphi=\all x:A. \psi\mapsto \all x:A. \psi _1= \vphi_1 \qquad \qquad \vphi =\all x:A. \psi\mapsto \all x:A. \psi_2= \vphi_2\]But this means that 
\[ \psi\mapsto   \psi _1  \qquad  \qquad   \psi\mapsto \psi_2\]
which, by inductive hypothesis, implies that there exists a $\psi_3$ for which
\[\psi _1 \mapsto  \psi _3 \qquad \qquad  \psi _2 \mapsto  \psi_3\]This, in turn, implies that
\[\chot\psi _1 \mapsto  \chot\psi _3 \qquad \qquad  \chot\psi _2 \mapsto  \chot\psi_3\]which gives us what we were looking for if we take $\chot\psi_3 $ as our $\vphi_3$. Hence,  $\vphi = \all x:A. \psi$ verifies the statement.

\item $\vphi = \chot \psi$. 
%By inductive hypothesis, if $\psi \mapsto \psi _1 $ and $\psi \mapsto \psi _2 $, then there exists a $\psi_3$ for which $\psi _1 \mapsto  \psi _3$ and $\psi _2 \mapsto  \psi_3$. 

Now, clearly, $\chot \psi$ can only reduce to terms of the form $\chot \psi \mapsto \chot \psi'$ such that $\psi\mapsto  \psi'$ and, moreover, $\psi\mapsto  \psi'$  implies that $\chot \psi \mapsto \chot \psi'$. Therefore, any two reductions $\chot \psi=\vphi \mapsto \vphi_1 $ and $\chot \psi=\vphi \mapsto \vphi_2$ can be represented as follows: 
\[\vphi=\chot \psi\mapsto \chot \psi _1= \vphi_1 \qquad \qquad \vphi =\chot \psi\mapsto \chot \psi_2= \vphi_2\]But this means that 
\[ \psi\mapsto   \psi _1  \qquad  \qquad   \psi\mapsto \psi_2\]
which, by inductive hypothesis, implies that there exists a $\psi_3$ for which
\[\psi _1 \mapsto  \psi _3 \qquad \qquad  \psi _2 \mapsto  \psi_3\]This, in turn, implies that
\[\chot\psi _1 \mapsto  \chot\psi _3 \qquad \qquad  \chot\psi _2 \mapsto  \chot\psi_3\]which gives us what we were looking for if we take $\chot\psi_3 $ as our $\vphi_3$. Hence,  $\vphi = \chot \psi$ verifies the statement.

\item $\vphi = \psi \wedge \chi$. Now,  $\psi \wedge \chi$ can only reduce in one step to terms of the form $\psi ' \wedge \chi $ or $\psi \wedge \chi '$   such that $\psi\mapsto  \psi'$ and, respectively, $\chi\mapsto  \chi'$. Moreover, $\psi\mapsto  \psi'$  and $\chi\mapsto  \chi'$ imply that $\psi \wedge \chi \mapsto  \psi ' \wedge \chi $ and $\psi \wedge \chi \mapsto  \psi \wedge \chi '$. By inductive hypothesis,  $\psi$ and $\chi$ verify the statement. Therefore, if both reductions of $\vphi = \psi \wedge \chi$ act exclusively inside $\psi$ or exclusively inside $\chi$, the statement clearly holds for $\vphi = \psi\et \chi $ as well. Suppose then that $\psi \wedge \chi= \vphi \mapsto \vphi_1 = \psi ' \wedge \chi$ and $\psi \wedge \chi= \vphi \mapsto \vphi_2= \psi \wedge \chi '$. Since then we must have that $\chi\mapsto  \chi'$ and $\psi\mapsto  \psi'$, we also have that $ \psi ' \wedge \chi \mapsto  \psi ' \wedge \chi '$ and  $ \psi \wedge \chi ' \mapsto  \psi ' \wedge \chi '$. Taking $\psi ' \wedge \chi '$ as $\vphi _3$ gives us that the statement is verified by $\vphi $ as well.

\item $\vphi = \lambda x:A . \psi$. By inductive hypothesis, if $\psi \mapsto \psi _1 $ and $\psi \mapsto \psi _2 $, then there exists a $\psi_3$ for which $\psi _1 \mapsto  \psi _3$ and $\psi _2 \mapsto  \psi_3$.  Since $\lambda x:A .  \psi \mapsto \lambda x:A .  \psi'$ if, and only if, $\psi\mapsto  \psi'$ and since $\lambda x:A .  \psi$ cannot reduce in any other way, we have, by an argument analogous to the one used in the first case considered, that also $\vphi = \lambda x:A . \psi$ verifies the statement.

\item $\vphi = \psi t$. By inductive hypothesis, if $\psi \mapsto \psi _1 $ and $\psi \mapsto \psi _2 $, then there exists a $\psi_3$ for which $\psi _1 \mapsto  \psi _3$ and $\psi _2 \mapsto  \psi_3$.  Therefore, if both reductions of $\psi t $ exclusively act inside $\psi$, we have that that $\psi t \mapsto \psi _1 t $ and $\psi t \mapsto \psi _2 t $, and then $\psi _1 t \mapsto  \psi _3 t$ and $\psi _2 t \mapsto  \psi_3t $. If, otherwise, $\vphi = (\lambda x:A . \chi) t \mapsto \chi[t/x]=\vphi_1$ and $\vphi = (\lambda x:A . \chi) t \mapsto (\lambda x:A . \chi ') t =\vphi_2$, we take $\vphi_3 $ to be $\chi'[t/x]$. Indeed, clearly, $\vphi_2=(\lambda x:A . \chi ') t \mapsto \chi'[t/x]=\vphi_3$. Moreover,  by Lemma \ref{lem:reduction-substitution} and by possibly renaming the  variables which are bound in $\chi$, if $\chi \mapsto \chi'$ then $\vphi_1=\chi[t/x]\mapsto \chi'[t/x]=\vphi_3$. In conclusion, $\vphi = \psi t$ verifies the statement in all  cases.
\end{itemize}
\end{proof}

\section{Strong normalisation}
\label{sec:normalisation}

We generalise the proof of strong normalisation by Tait's reducibility technique presented in \cite[Chapter 6]{glt89}.

%When we state that $t:A$ for $A:*$, we assume that $A$ is a type in normal form. We show in Section \ref{sec:normalisation-constr} that also type constructors strongly normalise.

\begin{definition}[Normal and strongly normalising]
A term $t$ is normal if there is no term $t'$ and real number $p$ such that $t\mapsto _p t'$. A type constructor $\vphi $ is normal if there is no type constructor $\vphi ' $ such that $\vphi \mapsto \vphi '$.

A term $t$ (type constructor $\vphi $) is strongly normalising if there exists no infinite sequence of reduction steps  $t\mapsto _{p_1} t_1 \mapsto _{p_2}   t_2 \mapsto _{p_3}   t_3  \dots $ ($\vphi  \mapsto \vphi  _1 \mapsto \vphi  _2 \mapsto \vphi  _3  \dots $)
\end{definition}

We define now, for any type $A:*$ and kind $\Phi:\square$, the set $\red_A$ of reducible terms of type $A$ and $\red_{\Phi}$ of reducible type constructors of kind $\Phi$.  

%{\ehi SHOULD WE REQUIRE THAT TYPES USED TO DEFINE REDUCIBLE TERMS ARE NORMAL? I DO NOT THINK SO AND IT DOES NOT SEEM THAT WE NEED TO DO IT.}
%

\begin{definition}[Reducibles and sets of reducibles]For any type $A:*$ and kind $\Phi:\square$, the set $\red_A$ of reducibles of type $A$ and the set $\red _\Phi$ of reducibles of kind $\Phi$ are defined as follows:
\begin{itemize}
\item for $\Phi=*$, $\red_*$ contains all strongly normalising type constructors of kind $*$.

\item for $\Phi=\Pi x:B.\Psi$, $\red_{\Pi y:B.\Psi}$ contains all type constructors $\vphi  :\Pi y:B.\Psi$ such that, for any term $s:B$,  
%{\latent NONEED $s\in \red_A$},  
$\vphi  s\in \red_{\Psi[s/x]}$.

\item for $A=P$ atomic type, $\red_P$ contains all strongly normalising terms of type $A$.

\item for $A=B\et C$, $\red_{B\et C}$ contains all terms $t:B\et C$ such that $t\pi_0\in \red_B$ and $t\pi_1\in \red_C$.

%\item for $A=B\impl  C$,  $\red_{B\impl C}$ contains all terms $t:B\impl  C$ such that, for any term $s\in \red_B$,  $ts\in \red_C$.

\item for $A=\all y: B. C$,  $\red_{\all y: B. C}$ contains all terms $t:(\all y:B. C)$ such that, for any term $s\in \red_B$,  $ts\in \red_{C[s/y]}$.

\item for $A=\chot B$,  $\red_{\chot B}$ contains all terms $t:\chot B$ such that $t\nu \in \red_B$.

\end{itemize}
The set $\red$ of all reducibles is defined as follows:
\[(\bigcup_{A:*} \red_A) \cup ( \bigcup_{\Psi:\square} \red_{\Psi})\]
\end{definition}
Clearly, the definition of the set $\red_{\Pi y :B. \Psi}$ of all reducibles of kind $\Pi y :B. \Psi$ is modelled after the definition of the set $\red_{\all y:B.C}$ of all reducibles of type $\all y:B.C$. Nevertheless, since during the evaluation of type constructors we do not reduce the terms to which they are applied, there is no need to mention whether or not the arguments of type constructors of kind $\Pi y:B. \Psi$ must be reducible in order to obtain a reducible type constructor application. As we will show, all terms are reducible, but this does not influence in any way the evaluation termination proof for type constructors.

\begin{definition}[Neutral terms and type constructors]
A term or type constructor is neutral if it is in {\bf none} of the following forms:
\[\lan s,t \ran  \qquad \lambda x.t \qquad s\cho{p}t \qquad \lambda x . \vphi   \]
\end{definition}That is, $t$ is neutral if, and only if, each redex occurring in $ts, t\pi_i, t\nu $ either occurs in $t$ or occurs in $s$; and $\vphi $ is neutral if, and only if, each redex occurring in $\vphi t$ occurs in $\vphi $.

Let us introduce three conditions that we will show to hold for the sets of reducible terms. 
\begin{itemize}
\item[$\cruno$] 
\begin{itemize}
\item[$*$] If $t\in\red _A$, then $t$ is strongly normalising. 
\item[$\square$] If $\vphi \in\red _\Phi$, then $\vphi $ is strongly normalising.
\end{itemize}
\item[$\crdue$] 
\begin{itemize}
\item[$*$] If  $t:A$, $t':A'$, $t\in \red _A$ and $t\mapsto _p t'$, then $t' \in \red _{A'}$.
\item[$\square$] If $\vphi  \in \red _\Phi$ and $\vphi  \mapsto \vphi  '$, then $\vphi  ' \in \red _ \Phi$.
\end{itemize}
\item[$\crtre$] 
\begin{itemize}
\item[$*$] If 
\begin{itemize}
\item[--] $t  :A$ is neutral and 
\item[--] for each $t' :A'$ such that $t\mapsto _p t'$, we have that $t'\in \red _{ A'}$
\end{itemize}
then $t\in \red _ A$.
\item[$\square$] 
If \begin{itemize}
\item[--] $\vphi  $ is neutral and 
\item[--] for each $ \vphi  '$ such that $\vphi \mapsto \vphi  '$, we have that $\vphi  '\in \red _\Phi$
\end{itemize}
then $\vphi  \in \red _ \Phi $.
\end{itemize}
\end{itemize}
Where $A'$ is a type with the same propositional and quantificational structure as $A$ but that possibly contains a term which is the one-step reduced form of the corresponding term occurring in $A$.

Notice that $\crdue$ states that $\red$ is closed under evaluation (the relation $\mapsto _p$ on terms and $\mapsto$ on type constructors), and $\crtre$ states that $\red$ is closed under those instances of anti-evaluation (the inverse of $\mapsto _p$ and $\mapsto$) that yield neutral terms.

\begin{proposition}\label{prop:condred}
For any type $A$ and kind $\Phi$, $\red _A$ and $\red _\Phi$ verify $\cruno, \crdue$ and $\crtre$.
\end{proposition}
\begin{proof}We prove the statement by induction on the number of symbol occurrences in $A$ and in $\Phi$.
In the base case for types, $A=P$ is an atomic type. 

\begin{itemize}
\item $\cruno$. Since  $\red _P$ only contains strongly normalising terms by definition,  $\cruno $ is trivially verified. 
\item $\crdue$. Since all elements of $\red _P$ are strongly normalising by definition and since a strongly normalising term can only reduce to strongly normalising terms, $\crdue$ is verified. 
\item $\crtre$. We suppose that $t:P$ is neutral and that, for each $t'$ such that $t\mapsto _p t'$, we have $t'\in \red _{P'}$. We need to prove that $t\in \red _P$. Any sequence of reductions starting from $t$ must be of the form $t \mapsto _{p_1} t_1 \mapsto _{p_1} \dots $ where $t_1\in \red _{P_1} $ and thus $t_1$ is strongly normalising. Hence, any sequence of reductions starting from $t$ must be finite. But this means that $t:P$ is strongly normalising, which implies, by definition of $\red _P$, that $t\in \red _P$. Hence, also $\crtre$ is verified. 
\end{itemize}

In the base case for kinds, $\Phi=*$ is an atomic kind. 

\begin{itemize}
\item $\cruno$. Since  $\red _*$ only contains strongly normalising elements by definition,  $\cruno $ is trivially verified. 
\item $\crdue$. Since all elements of $\red _*$ are strongly normalising by definition and since a strongly normalising constructor can only reduce to strongly normalising constructors, $\crdue$ is verified. 
\item $\crtre$. We suppose that $\vphi :*$ is neutral and that, for each $\vphi '$ such that $\vphi\mapsto  \vphi '$, we have $\vphi '\in \red _*$. We need to prove that $\vphi\in \red _*$. Any sequence of reductions starting from $\vphi$ must be of the form $\vphi \mapsto  \vphi_1 \mapsto  \dots $ where $\vphi _1\in \red _* $ and thus $\vphi _1$ is strongly normalising. Hence, any sequence of reductions starting from $\vphi$ must be finite. But this means that $\vphi:*$ is strongly normalising, which implies, by definition of $\red _*$, that $\vphi\in \red _*$. Hence, also $\crtre$ is verified. 
\end{itemize}

Suppose now that, for any type $F$ and kind $\Psi$ containing less than $n$ symbols, $\red _F$ and $\red_\Psi$ verify $\cruno, \crdue$ and $\crtre$. We prove that this holds also for any type $A$ and kind $\Phi$ containing $n$ symbols. 

Suppose that $A=B\et C$.
\begin{itemize}
\item $\cruno$. If $t\in \red _{B\et C}$, then $t\pi_0\in \red _B$. Hence, and by induction hypothesis ($\cruno$), $t\pi_0$ is strongly normalising. Suppose, by reasoning indirectly, that an infinite reduction sequence of the form $t\mapsto _{p_1} t_1\mapsto _{p_2}t_2 \dots $ exists, then also an infinite reduction sequence of the form $t\pi_0\mapsto _{p_1} t_1\pi_0\mapsto _{p_2}t_2\pi_0 \dots $ exists. Which is absurd since $t\pi_0 $ is strongly normalising. Hence, $t$ is strongly normalising and $\cruno$ is verified.

\item $\crdue$. Since $t\in \red _{B\et C}$, we have, by definition of $\red _{B\et C}$, that $t\pi_0\in \red _B$ and $t\pi_1\in \red _C$. Moreover, if $t\mapsto _p t'$, then $t\pi_0\mapsto _p t'\pi_0$ and $t\pi_1\mapsto _p t'\pi_1$. From $t\pi_0\in \red _B$, $t\pi_1\in \red _C$, $t\pi_0\mapsto _p t'\pi_0$ and $t\pi_1\mapsto _p t'\pi_1$, by induction hypothesis ($\crdue$), we obtain that $t'\pi_0\in \red _{B'}$ and $t'\pi_1\in \red _{C'}$. From this and by definition of $\red _{B\et C}$, we have that $t'\in \red _{B'\et C'}$, and thus that $\crdue$ is verified.

\item $\crtre$. We suppose that $t:B\et C$ is neutral and that, for each $t'$ such that $t\mapsto _p t'$, we have $t'\in \red _{B'\et C'}$. We need to prove that $t\in \red _{B\et C}$. Since $t$ is neutral, $t\pi_0$ can only reduce to a term of the form $t'\pi_0$ such that $t\mapsto _p t'$. Since the local assumptions guarantee us that $t'\in \red _{B'\et C'}$, we also know that $t'\pi_0\in \red _{B'}$. But then, since $t\pi_0$ is neutral and can only reduce to terms in $\red _{B'} $, by inductive hypothesis ($\crtre$), we have that $t\pi_0\in \red_B$. Since the same argument guarantees that $t\pi_1\in \red_C$, we have that $t\in \red _{B\et C}$ and thus that $\crtre$ is verified. 
\end{itemize}

Suppose that $A=\all y:B. C$.
\begin{itemize}
\item $\cruno$. If $t\in \red _{\all y:B. C}$, then $tx\in \red _{C[x/y]}$ for any $x: B$. Indeed, by inductive hypothesis ($\crtre$), from $x:A$ being normal and neutral we can infer that $x\in \red _A$. By induction hypothesis ($\cruno$), $tx\in \red _{C[x/y]}$ implies that $tx$ is strongly normalising. Suppose now, by reasoning indirectly, that an infinite reduction sequence of the form $t\mapsto _{p_1} t_1\mapsto _{p_2}t_2 \dots $ exists, then also an infinite reduction sequence of the form $tx\mapsto _{p_1} t_1x\mapsto _{p_2}t_2x \dots $ exists. Which is absurd since $tx$ is strongly normalising. Hence, $t$ is strongly normalising and $\cruno$ is verified.

\item $\crdue$. Since $t\in \red _{\all y:B. C}$, we have, by definition of $\red _{\all y:B. C}$, that $ts\in \red _{C[s/y]}$ for any $s\in \red_B$. Moreover, if $t\mapsto _p t'$, then $ts\mapsto _p t's$. From $ts \in \red _{C[s/y]}$ and $ts\mapsto _p t's$, by induction hypothesis ($\crdue$), we obtain that $t's\in \red _{C'[s/y]}$. Since $s$ is a generic element of $\red_B$, from $t's\in \red _{C'[s/y]}$ and by definition of $\red _{\all y:B. C'}$, we have that $t'\in \red _{\all y:B. C'}$, and thus that $\crdue$ is verified.

\item $\crtre$. We suppose that $t:(\all y:B. C)$ is neutral and that, for each $t'$ such that $t\mapsto _p t'$, we have $t'\in \red _{\all y:B. C'}$. We need to prove that $t\in \red _{\all y:B. C}$. In order to show this, it is enough to  show that $ts \in \red _{C[s/y]}$ for any $s\in \red _B$. Since, by induction hypothesis ($\crdue$),  $s\in \red _B$ implies that $s$ is strongly normalising, we know that the longest sequence of reductions starting from $s$ is finite and we can reason by induction on its length $l(s)$. 

If $l(s)=0$, since $t$ is neutral, then $ts$ can only reduce to terms of the form $t's$ such that $t\mapsto _p t'$. Since the local assumptions guarantee us that $t'\in \red _{\all y:B. C'}$, we also know that $t's\in \red _{C'[s/y]}$. But then, since $ts$ is neutral and can only reduce to terms in $\red _{C'[s/y]}$, by inductive hypothesis ($\crtre$), we have that $ts\in \red_{C[s/y]}$ as desired. 

Suppose now that $ts \in \red _{C[s/y]}$ for any $s\in \red _B$ such that $l(s)<n$, we show that this holds also for any $s\in \red _B$ such that $l(s)=n$. 

Since $t$ is neutral, $ts$ can either reduce to a term of the form $t's$ such that $t\mapsto _p t'$ or of the form $ts'$ such that $s\mapsto _p s'$. Since the local assumptions guarantee us that $t'\in \red _{\all y:B.C'}$, we also know that $t's\in \red _{C'[s/y]}$. We moreover have that $l(s')<l(s)$, and hence, by local inductive hypothesis, we know that $ts'\in \red _{C[s'/y]}$. But then, since $ts$ is neutral and can only reduce to terms in $\red$ (either in $\red _{C'[s/y]}$ or in $\red _{C[s'/y]}$), by global inductive hypothesis ($\crtre$), we have that $ts\in \red _{C[s/y]}$.

Since we just showed that $ts\in \red _{C[s/y]}$ for any $s\in \red _B$, we have that $t\in \red _{\all y:B.C}$, which means that $\crtre$ is verified.
\end{itemize}

Suppose that $A=\chot B$.
\begin{itemize}
\item $\cruno$. If $t\in \red _{\chot B}$, then $t\nu\in \red _B$. Hence, and by induction hypothesis ($\cruno$), $t\nu$ is strongly normalising. Suppose, by reasoning indirectly, that an infinite reduction sequence of the form $t\mapsto _{p_1} t_1\mapsto _{p_2}t_2 \dots $ exists, then also an infinite reduction sequence of the form $t\nu\mapsto _{p_1} t_1\nu\mapsto _{p_2}t_2\nu \dots $ exists. Which is absurd since $t\nu $ is strongly normalising. Hence, $t$ is strongly normalising and $\cruno$ is verified.

\item $\crdue$. Since $t\in \red _{\chot B}$, we have, by definition of $\red _{\chot B}$, that $t\nu\in \red _B$. Moreover, if $t\mapsto _p t'$, then $t\nu\mapsto _p t'\nu$. From $t\nu\in \red _B$ and $t\nu\mapsto _p t'\nu$, by induction hypothesis ($\crdue$), we obtain that $t'\nu\in \red _{B'}$. From this and by definition of $\red _{\chot B}$, we have that $t'\in \red _{\chot B'}$, and thus that $\crdue$ is verified.

\item $\crtre$. We suppose that $t:\chot B$ is neutral and that, for each $t'$ such that $t\mapsto _p t'$, we have $t'\in \red _{\chot B'}$. We need to prove that $t\in \red _{\chot B}$. Since $t$ is neutral, $t\nu$ can only reduce to a term of the form $t'\nu$ such that $t\mapsto _p t'$. Since the local assumptions guarantee us that $t'\in \red _{\chot B'}$, we also know that $t'\nu\in \red _{B'}$. But then, since $t\nu$ is neutral and can only reduce to terms in $\red _{B'}$, by inductive hypothesis ($\crtre$), we have that $t\nu\in \red_B$. Hence, by definition of $\red_{\chot B}$, we have that $t\in \red _{\chot B}$ and thus that $\crtre$ is verified. 
\end{itemize}

Suppose that $\Phi=\Pi y:B. \Xi$.
\begin{itemize}
\item $\cruno$. If $\vphi\in \red _{\Pi y:B. \Xi}$, then $\vphi x\in \red _{\Xi[x/y]}$ for any $x: B$. 
%{\latent NONEED Indeed, by inductive hypothesis ($\crtre$), from $x:A$ being normal and neutral we can infer that $x\in \red _A$.} 
By induction hypothesis ($\cruno$), $\vphi x\in \red _{\Xi[x/y]}$ implies that $\vphi x$ is strongly normalising. Suppose now, by reasoning indirectly, that an infinite reduction sequence of the form $\vphi\mapsto \vphi_1\mapsto \vphi_2 \dots $ exists, then also an infinite reduction sequence of the form $\vphi x\mapsto \vphi _1x\mapsto \vphi _2x \dots $ exists. Which is absurd since $\vphi x$ is strongly normalising. Hence, $t$ is strongly normalising and $\cruno$ is verified.

\item $\crdue$. Since $\vphi \in \red _{\Pi y:B. \Xi}$, we have, by definition of $\red _{\Pi y:B. \Xi}$, that $\vphi s\in \red _{\Xi[s/y]}$ for any 
$s:B$. 
%{\latent NONEED $s\in \red_B$}. 
Moreover, if $\vphi \mapsto \vphi '$, then $\vphi s\mapsto \vphi 's$. From $\vphi s \in \red _{\Xi[s/y]}$ and $\vphi s\mapsto \vphi 's$, by induction hypothesis ($\crdue$), we obtain that $\vphi 's\in \red _{\Xi[s/y]}$. Since $s$ is a generic term of type $B$, 
%{\latent NONEED  element of $\red_B$}, 
 from $\vphi 's\in \red _{\Xi[s/y]}$ and by definition of $\red _{\Pi y:B. \Xi }$, we have that $\vphi '\in \red _{\Pi y:B. \Xi }$, and thus that $\crdue$ is verified.

\item $\crtre$. We suppose that $\vphi :(\Pi y:B. \Xi)$ is neutral and that, for each $\vphi '$ such that $\vphi \mapsto _p \vphi '$, we have $\vphi '\in \red _{ \Pi y:B. \Xi }$. We need to prove that $\vphi \in \red _{\Pi y:B. \Xi}$. In order to show this, it is enough to  show that $\vphi s \in \red _{\Xi[s/y]}$ for any
$s:B$.
%{\latent NONEED $s\in \red _B$}.
Since $\vphi $ is neutral, then $\vphi s$ can only reduce to terms of the form $\vphi 's$ such that $\vphi \mapsto _p \vphi '$. Now, the local assumptions guarantee us that $\vphi '\in \red _{ \Pi y:B. \Xi}$. Therefore, we also know that $\vphi 's\in \red _{\Xi[s/y]}$. But then, since $\vphi s$ is neutral and can only reduce to terms in $\red _{\Xi[s/y]}$, by inductive hypothesis ($\crtre$), we have that $\vphi s\in \red_{\Xi[s/y]}$ as desired. Considering that we just showed that $\vphi s\in \red _{\Xi[s/y]}$ for any $s:B$, 
%{\latent NONEED $s\in \red _B$}, 
we have that $\vphi \in \red _{\Pi y:B. \Xi}$, which means that $\crtre$ is verified.

\end{itemize}

\end{proof}

We now show that the term constructors corresponding to logical introduction rules preserve reducibility. 
\begin{lemma}[Pairing lemma]\label{lem:pairing}
If $s\in \red _A $ and $t\in \red_B$ then $\lan s,t\ran\in \red _{A\et B}$.
\end{lemma}
\begin{proof}In order to prove the statement, we need to show that 
$\lan s,t\ran\pi_0\in \red _A$ and $\lan s,t\ran\pi_1\in \red _B$.

Since we have that $s\in \red _A $ and $t\in \red_B$ and because of $\cruno$, we know that all sequences of reductions starting form $s$ and $t$ are of finite length. Hence, we can prove the statement by induction on $l(s)+l(t)$ where the function $l(\;)$ has as value the length of the longest sequence of reductions starting from its argument.

Let us consider $\lan s,t\ran\pi_0:A$. In one step, it can only reduce to one of the following terms:
\begin{itemize}
\item $s\in \red _A$
\item $\lan s',t\ran\pi_0$ where $s\mapsto _p s'$. Since, by $\crdue $, $s'\in \red_{A'}$ and since $l(s')<l(s)$, we can apply the inductive hypothesis to $\lan s',t\ran\pi_0$ and conclude that $\lan s',t\ran\pi_0\in \red _{A'}$.  
 
\item $\lan s,t'\ran\pi_0$ where $t\mapsto _p t'$. Since, by $\crdue $, $t'\in \red_{B'}$ and since $l(t')<l(t)$, we can apply the inductive hypothesis to $\lan s,t'\ran\pi_0$ and conclude that $\lan s,t'\ran\pi_0\in \red _{A}$.  
\end{itemize}
Hence, in one step, $\lan s,t\ran\pi_0$  can only reduce to terms in $\red$ (either in $\red _A$ or in $\red _{A'}$). By $\crtre$ we then know that it is an element of $\red _A$.

Since we can conclude by an analogous argument that $\lan s,t\ran\pi_1\in \red _B$, we have that $\lan s,t\ran \in \red _{A\et B}$, as desired.
\end{proof}

%\begin{lemma}[Abstraction lemma]\label{lem:abstraction}
%If, for any $s\in \red _A$, $t[s/x]\in \red _B $, then $\lambda x.t \in \red _{A\impl  B}$.
%\end{lemma}
%\begin{proof}In order to prove the statement, we need to show that, for any $s \in \red_ A$, $(\lambda x.t)s \in \red _B$. Since, by Proposition \ref{prop:condred} ($\cruno$), $s \in \red_ A$ and $t \in \red _{A\impl B}$ imply that $s$ and $t$ are strongly normalising, we can reason by induction on $l(s)+l(t)$, where $l(\;)$ has as value the length of the longest sequence of reductions starting from its argument.
%
%Let us consider the possible outcomes of a one-step reduction of $(\lambda x.t)s$:
%\begin{itemize}
%\item $t[s/x]$, which is assumed to be an element of $\red_B$. 
%\item $(\lambda x.t')s$. Since $t\mapsto _p t'$, $l(s)+l(t')<l(s)+l(t)$. We can then use the induction hypothesis to conclude that $(\lambda x.t')s\in \red _B$.
%\item $(\lambda x.t)s'$. Since $s\mapsto _p s'$, $l(s')+l(t)<l(s)+l(t)$. We can then use the induction hypothesis to conclude that $(\lambda x.t)s'\in \red _B$.
%\end{itemize}
%Since all terms to which $(\lambda x.t)s:B$ can reduce in one step are elements of $\red _B$, by Proposition \ref{prop:condred} ($\crdue$), we know that $(\lambda x.t)s\in \red _B$. Since, moreover, $s$ is a generic element of $\red_B$, by definition of $\red _{\all y:B. C}$, we can conclude, as desired, that $(\lambda x.t)\in \red _{A\impl B}$.
%\end{proof}

\begin{lemma}[Abstraction lemma]\label{lem:abstraction}
If, for any $s\in \red _A$, $t[s/x]\in \red _{B[s/x]} $, then $\lambda x.t \in \red _{\all x: A. B}$.
\end{lemma}
\begin{proof}In order to prove the statement, we need to show that, for any $s \in \red_ A$, $(\lambda x.t)s \in \red _{B[s/x]}$. Since, by Proposition \ref{prop:condred} ($\cruno$), $s \in \red_ A$ and $t \in \red _{\all x: A. B}$ imply that $s$ and $t$ are strongly normalising, we can reason by induction on $l(s)+l(t)$, where $l(\;)$ has as value the length of the longest sequence of reductions starting from its argument.

Let us consider the possible outcomes of a one-step reduction of $(\lambda x.t)s$:
\begin{itemize}
\item $t[s/x]$, which is assumed to be an element of $\red_{B[s/x]}$. 
\item $(\lambda x.t')s$. Since $t\mapsto _p t'$, $l(s)+l(t')<l(s)+l(t)$. We can then use the induction hypothesis to conclude that $(\lambda x.t')s\in \red _B'[s/x]$.
\item $(\lambda x.t)s'$. Since $s\mapsto _p s'$, we also have that $l(s')+l(t)<l(s)+l(t)$. We can then use the induction hypothesis to conclude that $(\lambda x.t)s'\in \red _B[s'/x]$.
\end{itemize}
Since all terms to which $(\lambda x.t)s:B[s/x]$ can reduce in one step are elements of $\red $ (of $\red _{B[s/x]}$, of $B'[s/x]$ or of$B[s'/x]$), by Proposition \ref{prop:condred} ($\crdue$), we know that $(\lambda x.t)s\in \red _{B[s/x]}$. Since, moreover, $s$ is a generic element of $\red_B$, by definition of $\red _{\all x:B. C}$, we can conclude, as desired, that $(\lambda x.t)\in \red _{\all x: A. B}$.
\end{proof}

\begin{lemma}[Nondeterministic choice lemma]\label{lem:ndchoice}
If $s, t\in \red_A$ then $ s\cho{p}t \in \red _{\chot A}$.
\end{lemma}
\begin{proof}In order to prove the statement, we need to show that $(s\cho{p}t)\nu  \in \red _A$.

Since we have that $s,t\in \red_A$ and because of $\cruno$, we know that all sequences of reductions starting form $s$ and $t$ are of finite length. Hence, we can prove the statement by induction on $l(s)+l(t)$ where the function $l(\;)$ has as value the length of the longest sequence of reductions starting from its argument. 

Let us consider $(s\cho{p}t)\nu:A$. In one step, it can only reduce to one of the following terms:
\begin{itemize}
\item $s\in \red _A$
\item $t\in \red _A$
\item $(s'\cho{p}t)\nu$ where $s\mapsto _p s'$. Since, by $\crdue $, $s'\in \red_{A'}$ and since $l(s')<l(s)$, we can apply the inductive hypothesis to $(s'\cho{p}t)\nu$ and conclude that $(s'\cho{p}t)\nu\in \red _{A'}$.  
 
\item $(s\cho{p}t')\nu$ where $t\mapsto _p t'$. Since, by $\crdue $, $t'\in \red_{A'}$ and since $l(t')<l(t)$, we can apply the inductive hypothesis to $(s\cho{p}t')\nu$ and conclude that $\lan s,t'\ran\pi_0\in \red _{A'}$.  
\end{itemize}
Hence, in one step, $(s\cho{p}t)\nu$ can only reduce to terms in $\red$ (either in $\red _{A}$ or in $\red _{A'}$). By $\crtre$ we then know that it is an element of $\red _A$.

We have then that $s\cho{p}t \in \red _{\chot A}$, as desired.
\end{proof}

\begin{lemma}[Abstraction lemma for type constructors]\label{lem:abstraction-tc}
If, for any 
$s:A$, 
%{\latent NONEED $s\in \red _A$}, 
$\vphi [s/x]\in \red _{\Xi[s/x]} $, then $\lambda x.\vphi  \in \red _{\Pi x: A. \Xi}$.
\end{lemma}
\begin{proof}In order to prove the statement, we need to show that, for any 
$s:A$, 
%{\latent NONEED $s \in \red_ A$}, 
$(\lambda x. \vphi )s \in \red _{\Xi[s/x]}$. Since, by Proposition \ref{prop:condred} ($\cruno$), $\vphi  \in \red _{\Pi x: A. \Xi}$ implies that $\vphi $ is strongly normalising, we can reason by induction on $l(\vphi )$, where $l(\;)$ has as value the length of the longest sequence of reductions starting from its argument.

Let us consider the possible outcomes of a one-step reduction of $(\lambda x.\vphi )s$:
\begin{itemize}
\item $\vphi [s/x]$, which is assumed to be an element of $\red_{\Xi[s/x]}$. 
\item $(\lambda x.\vphi ')s$. Since $\vphi \mapsto _p \vphi '$, we have also $l(\vphi ')<l(\vphi )$. We can then use the induction hypothesis to conclude that $(\lambda x.\vphi ')s\in \red _\Xi[s/x]$.
\end{itemize}
Since all type constructors to which $(\lambda \vphi )s:\Xi[s/x]$ can reduce in one step are elements of $\red _{\Xi[s/x]}$, by Proposition \ref{prop:condred} ($\crdue$), we know that $(\lambda x.\vphi )s\in \red _{\Xi[s/x]}$. Since, moreover, $s$ is a generic 
term of type $B$, 
%{\latent NONEED element of $\red_B$}, 
by definition of $\red _{\Pi x:B. \Xi}$, we can conclude, as desired, that $(\lambda x.\vphi )\in \red _{\Pi x: A. \Xi}$.
\end{proof}

\begin{theorem}[Reducibility theorem]\label{thm:red}
For any type $A:*$ and term $t:A$, $t\in \red_A$.
\end{theorem}
\begin{proof}We prove the following stronger statement:
\begin{quote}
For any type $F:*$, term $t:F$ with free variables  $x_1:A_1, \dots , x_n :A_n$ and sequence of terms $s_1 \in \red _{A_1}, \dots , s_n \in \red _{A_n}$, we have that $t[s_1/x_1 , \dots , s_n /x_n ]\in \red _{F[s_1/x_1 , \dots , s_n /x_n ]}$.
\end{quote} 
Let us denote by $[\overline{s}/\overline{x}]$ the substitution $[s_1/x_1 , \dots , s_n /x_n ]$. The proof is by induction on the number of symbol occurrences in $t$.
\begin{itemize}
\item $t=x_i$ for some $1\leq i \leq n$. Since $t[\overline{s}/\overline{x}]=s_i$ and $s_i\in \red _{A_i}$ by assumption, we obviously have $t[\overline{s}/\overline{x}] \in \red _{F[\overline{s}/\overline{x}]}$ where $F[\overline{s}/\overline{x}]=A_i$.  

\item $t=\lan u,v\ran$. By induction hypothesis, we have that $u[\overline{s}/\overline{x}] \in \red _{F_1[\overline{s}/\overline{x}]}$ and $v[\overline{s}/\overline{x}] \in \red _{F_2[\overline{s}/\overline{x}]}$ for suitable types $F_1, F_2$. By Lemma \ref{lem:pairing}, we know that $\lan u[\overline{s}/\overline{x}] , v[\overline{s}/\overline{x}] \ran \in \red _{F_1[\overline{s}/\overline{x}]\et F_2[\overline{s}/\overline{x}]}$. But $\lan u[\overline{s}/\overline{x}] , v[\overline{s}/\overline{x}] \ran= \lan u,v\ran[\overline{s}/\overline{x}]$ and $F_1[\overline{s}/\overline{x}]\et F_2[\overline{s}/\overline{x}]=(F_1\et F_2)[\overline{s}/\overline{x}]$, and hence $\lan u,v\ran[\overline{s}/\overline{x}] \in \red _{(F_1\et F_2)[\overline{s}/\overline{x}]}$ as desired. 

\item $t=u\pi_0$ and $t=u\pi_1$. By induction hypothesis, $u[\overline{s}/\overline{x}]\in \red _{(F_1\et F_2)[\overline{s}/\overline{x}]}$ for any sequence $\overline {u}=s_1, \dots , s_n$ such that $s_1 \in \red _{A_1}, \dots , s_n \in \red _{A_n}$. By definition of $\red _{(F_1\et F_2)[\overline{s}/\overline{x}]}$ and since $(F_1\et F_2)[\overline{s}/\overline{x}]=F_1[\overline{s}/\overline{x}]\et F_2[\overline{s}/\overline{x}]$, this means that $(u[\overline{s}/\overline{x}])\pi_i\in \red _{F_i[\overline{s}/\overline{x}]}$ for $i\in \{0,1\}$. But since $(u[\overline{s}/\overline{x}])\pi_i = (u\pi_i)[\overline{s}/\overline{x}]$, we have that $ (u\pi_i)[\overline{s}/\overline{x}] \in \red _{F_i[\overline{s}/\overline{x}]}$, as desired.  
 
\item $t=\lambda y. u$. By induction hypothesis, for suitable types $F_1$ and $F_2$, we have that $u[\overline{s}/\overline{x}, v/y] \in \red _{F_2[\overline{s}/\overline{x}, v/y]}$ for any $v\in \red _{F_1}$. By Lemma \ref{lem:abstraction}, this implies that $\lambda y. (u[\overline{s}/\overline{x}])\in \red _{\all y:F_1.(F_2[\overline{s}/\overline{x}])}$. But the assumption that $\overline{x}$ are all free in $t=\lambda y. u$ implies that $y\notin \overline{x}$, which in turn implies that  $\lambda y. (u[\overline{s}/\overline{x}])= (\lambda y. u)[\overline{s}/\overline{x}]$ and that $\all y:F_1.(F_2[\overline{s}/\overline{x}]) = (\all y:F_1.F_2)[\overline{s}/\overline{x}]$. Hence, $ (\lambda y. u)[\overline{s}/\overline{x}] \in \red _{(\all y:F_1.F_2)[\overline{s}/\overline{x}]}$, as desired.

\item $t=uv$. By induction hypothesis, $u[\overline{s}/\overline{x}]\in \red _{(\all y:F_1. F_2)[\overline{s}/\overline{x}]}$ and $v[\overline{s}/\overline{x}]\in \red _{F_1[\overline{s}/\overline{x}]}$ for any sequence $\overline {s}=s_1, \dots , s_n$ such that $s_1 \in \red _{A_1}, \dots , s_n \in \red _{A_n}$. Since $ (\all y:F_1.F_2)[\overline{s}/\overline{x}]=\all y:(F_1[\overline{s}/\overline{x}]).(F_2[\overline{s}/\overline{x}]) $ and by definition of $\all y:(F_1[\overline{s}/\overline{x}]).(F_2[\overline{s}/\overline{x}]) $, this means that $(u[\overline{s}/\overline{x}])(v[\overline{s}/\overline{x}])\in \red _{F_2[\overline{s}/\overline{x}]}$. But since $(u[\overline{s}/\overline{x}])(v[\overline{s}/\overline{x}]) = (uv)[\overline{s}/\overline{x}]$, we have that $(uv)[\overline{s}/\overline{x}] \in \red _{F_2[\overline{s}/\overline{x}]}$, as desired.

\item $t= u\cho{p}v$. By induction hypothesis, we have that $u[\overline{s}/\overline{x}] , v[\overline{s}/\overline{x}] \in \red _{F_1[\overline{s}/\overline{x}]}$ for a suitable type $F_1$. By Lemma \ref{lem:ndchoice}, we know that $ u[\overline{s}/\overline{x}] \cho{p} v[\overline{s}/\overline{x}]  \in \red _{\chot F_1[\overline{s}/\overline{x}]}$. But $ u[\overline{s}/\overline{x}] \cho{p} v[\overline{s}/\overline{x}]   = (u \cho{p} v)[\overline{s}/\overline{x}]  $ and $\chot F_1[\overline{s}/\overline{x}]=(\chot F_1)[\overline{s}/\overline{x}]$. Hence, $ (u \cho{p} v)[\overline{s}/\overline{x}] \in \red _{(\chot F_1)[\overline{s}/\overline{x}]}$ as desired.

\item $t=u\nu$. By induction hypothesis, $u[\overline{s}/\overline{x}]\in \red _{(\chot F_1)[\overline{s}/\overline{x}]}$ for any sequence $\overline {u}=s_1, \dots , s_n$ such that $s_1 \in \red _{A_1}, \dots , s_n \in \red _{A_n}$. Since $(\chot F_1)[\overline{s}/\overline{x}]=\chot F_1[\overline{s}/\overline{x}]$ and by definition of $\red _{\chot F_1[\overline{s}/\overline{x}]}$, this means that $(u[\overline{s}/\overline{x}])\nu\in \red _{F_1[\overline{s}/\overline{x}]}$. But since $(u[\overline{s}/\overline{x}])\nu = (u\nu)[\overline{s}/\overline{x}]$, we have that $ (u\nu)[\overline{s}/\overline{x}] \in \red _{F_1[\overline{s}/\overline{x}]}$, as desired.  
\end{itemize}
\end{proof}

\begin{theorem}[Reducibility theorem for type constructors]\label{thm:red-tc}
For any kind $\Phi:\square$ and type constructor $\vphi :\Phi$, $\vphi \in \red_\Phi$.
\end{theorem}
\begin{proof}We prove the following stronger statement:
\begin{quote}
For any  kind $\Phi:\square$ and type constructor $\vphi :\Phi$ with free variables $x_1:A_1, \dots , x_n :A_n$ and sequence of terms 
$s_1: A_1, \dots , s_n :A_n$, 
%{\latent NONEED $s_1 \in \red _{A_1}, \dots , s_n \in \red _{A_n}$}, 
we have that $\vphi [s_1/x_1 , \dots , s_n /x_n ]\in \red _{\Phi[s_1/x_1 , \dots , s_n /x_n ]}$.
\end{quote} 
Let us denote by $[\overline{s}/\overline{x}]$ the substitution $[s_1/x_1 , \dots , s_n /x_n ]$. The proof is by induction on the number of symbol occurrences in $\vphi$.
\begin{itemize}
\item $\vphi  = \alpha$. Now, if $\vphi[\overline{s}/\overline{x}]=\alpha$ then $  \vphi[\overline{s}/\overline{x}]=\vphi$. Moreover, since $\alpha :\Phi$ is normal and neutral we can infer by $\crtre$ that $\alpha \in \red _\Phi$. Therefore, $\vphi \in \red _\Phi$.  

\item $\vphi =\lambda y. \psi$. By induction hypothesis, for a suitable kind $\Phi$ and type $F$, we have that $\psi [\overline{s}/\overline{x}, v/y] \in \red _{\Phi[\overline{s}/\overline{x}, v/y]}$ for any 
$v:F$. 
%{\latent NONEED $v\in \red _{F}$}. 
By Lemma \ref{lem:abstraction-tc}, this implies that $\lambda y. (\psi[\overline{s}/\overline{x}])\in \red _{\Pi y:F.(\Phi[\overline{s}/\overline{x}])}$. But the assumption that $\overline{x}$ are all free in $\vphi =\lambda y. \psi$ implies that $y\notin \overline{x}$, which in turn implies that  $\lambda y. (\psi[\overline{s}/\overline{x}])= (\lambda y. \psi)[\overline{s}/\overline{x}]$ and that $\Pi y:F.(\Phi[\overline{s}/\overline{x}]) = (\Pi y:F.\Phi)[\overline{s}/\overline{x}]$. Hence, $ (\lambda y. \psi)[\overline{s}/\overline{x}] \in \red _{(\Pi y:F.\Phi)[\overline{s}/\overline{x}]}$, as desired.

\item $\vphi =\psi v$. By induction hypothesis, $\psi [\overline{s}/\overline{x}]\in \red _{(\all y:F. \Phi)[\overline{s}/\overline{x}]}$ and 
$v[\overline{s}/\overline{x}]: F[\overline{s}/\overline{x}]$ 
%{\latent NONEED AND NOHAVE (WE HAVE IT IF WE MERGE THE TWO PROOFS) $v[\overline{s}/\overline{x}]\in \red _{F[\overline{s}/\overline{x}]}$ } 
for any sequence $\overline {s}=s_1, \dots , s_n$ such that $s_1 \in \red _{A_1}, \dots , s_n \in \red _{A_n}$. Since $ (\Pi y:F.\Phi)[\overline{s}/\overline{x}]=\Pi y:(F[\overline{s}/\overline{x}]).(\Phi[\overline{s}/\overline{x}]) $, by definition of $\Pi y:(F[\overline{s}/\overline{x}]).(\Phi[\overline{s}/\overline{x}]) $, we have that $(\psi [\overline{s}/\overline{x}])(v[\overline{s}/\overline{x}])\in \red _{\Phi[\overline{s}/\overline{x}]}$. But since $(\psi [\overline{s}/\overline{x}])(v[\overline{s}/\overline{x}]) = (\psi v)[\overline{s}/\overline{x}]$, we obtain $(\psi v)[\overline{s}/\overline{x}] \in \red _{\Phi[\overline{s}/\overline{x}]}$, as desired.  

\end{itemize}
\end{proof}

\begin{corollary}[Strong normalisation]
For any type $A:*$ and term $t:A$, $t$ is strongly normalising. For any kind $\Phi:\square$ and type constructor $\vphi :\Phi $, $\vphi $ is strongly normalising.
\end{corollary}
\begin{proof}We just need to apply Theorems \ref{thm:red} and \ref{thm:red-tc} with $s_i =x_i$ for each $i\in \{1\dots, n\}$. The result then follows by $\cruno$.\end{proof}

\section{A theory of program evaluation and trustworthiness}\label{sec:conclusion}

We present now an extension of the system that aims at providing a formal framework for reasoning about computational trustworthiness of deterministic and nondeterministic programs. Indeed, the extensive deployment of probabilistic algorithms has radically changed our perspective on several well-established computational notions. {\it Correctness} is probably the most basic one. While a typical probabilistic program cannot be said to compute {\it the} correct result, we often have quite strong expectations about the frequency with which it should return certain outputs. Instead of talking about correctness, then, it is possible to introduce a more general notion that captures this intuition: let us call a probabilistic computational process {\it trustworthy} if the frequency of its outputs is compliant with a probability distribution which models its expected behaviour.

%The conversion rules, Table \ref{tab:type-conv-rules}, enable us to evaluate type constructor applications during the typing phase. This is required since type constructors are meant, not surprisingly, to produce types that can then be assigned to terms. At the same time, type constructors produce types by being applied to arguments and by being evaluated, not unlike terms produce outputs. This generates an apparent contradiction: we need the type constructor to provide us with a type to be used while typing, but the type constructor can only this type during its evaluation, which happens after the typing phase is completed. The adopted solution consists in enabling us to evaluate type constructor applications during the typing phase by the evaluation rules. While we deductively construct our term, we can also construct its type by forming the required type constructors, by applying them and by evaluating them, and we can do all this inside the derivation through which we are constructing our term. Thus we do not need to wait for the dynamic phase of evaluation in order to obtain the type we need to assign to our term.

The idea that we will employ to define a {\it trustworthiness} predicate $\trust$ that formalises this intuition is similar to the one at the basis of the conversion rules in Table \ref{tab:type-conv-rules}. Indeed, trust is a relation that essentially depends on our experience of the past behaviour of a program during evaluation. The more often we experience that the program meets our expectations with respect to the frequency of its outputs, the more we will be inclined to trust the program. Hence, the definition of trust must refer to the evaluation---or better, a number of past evaluations---of the program. Since we still want to be able to apply the trust predicate during the typing phase, and not to wait for the end of the dynamic evaluation phase, we also include in the system evaluation rules corresponding to certain reductions of terms. Thus we will be able to study during the typing phase the runtime behaviour of terms and establish statically whether or not we trust the terms. 
%
%Notice that checking whether an application of a rule for $\downarrow $ is correct is easier than to check whether an application of the rule for the evaluation of type constructors is correct. Indeed, for the first kind of rules, it is enough to check whether a term reduces in one step into another, while for the second rule, in general, we need to check whether a type constructor reduces to another one, possibly in several reduction steps.
%
%

%We then define on the basis of it a computational trust predicate in the style of \cite{gp23}. The novelty of this approach with respect to  \cite{gp23} lies in the fact that the presented calculus does not only enable us to precisely formalise a notion of trust and to encode the procedures required to test it on programs but also to reason by means of the type system on the behaviour of programs and on 

In order to statically study the runtime behaviour of programs, we introduce the evaluation predicate $\Mapsto$. Before presenting its rules, let us then define some useful notions.

\begin{definition}[Static reduction sequence]\label{def:seq-static-red}
A sequence  $r_1 , \dots , r_n $ of static reductions is a list of quadruples of the form $(\calc[s], \calc[s'], q , \rho)$ such that, for $1\leq i\leq n $, the second element of $r_i$ and the first element of $r_{i+1}$ are the same term and, moreover, for each quadruple $r_1 , \dots , r_n $, the following holds:
\begin{itemize} 
\item $\calc[s]\mapsto _q \calc[s']$ according to Table \ref{tab:evaluation},
\item $\calc[s]$ is of the suitable form to be displayed to the right of $:$ in the premiss of one of the evaluation rules in Table \ref{tab:eval-pred-rules} with $\calc[s']$ displayed, in the conclusion of the rule, to the right of $:$,  
\item $\rho$ is an element of the set of labels $\{\beta , \textit{left}, \textit{right}, \omega\}$ such that
\begin{itemize}
\item $\rho=\beta$ when $s=(\lambda x.u)v \mapsto _1 u[v/x]=s'$,
\item $\rho= \textit{left}$ when $s= (u\cho{p}v)\nu \mapsto _p u=s'$, 
\item $\rho=\textit{right} $ when $s= (u\cho{p}v)\nu \mapsto _{1-p} v=s'$, and 
\item $\rho=\omega$ if either $s=o\nu$ or $s=(ou)\nu$ for an oracle constant $o$.  
\end{itemize}
\end{itemize}
\end{definition}

\begin{definition}[Sequence of terms produced by a static reduction sequence]
\label{def:prod-seq-static-red}A sequence of terms $t_1 , \dots , t_n$ is produced by the static reduction sequence $r_1 , \dots , r_{n-1} $ if, and only if,  $t_n$ is the second element of $r_{n-1}$, and, for $1\leq i\leq n-1$, $t_i$ is the first element of $r_i$.
\end{definition}

\begin{definition}[$\not \equiv_{\mathrm{ND}}$]
For any two sequences of terms $t_1 , \dots , t_n$ and $s_1, \dots s_m$, we say that $t_1 , \dots , t_n\not\equiv_{\mathrm{ND}} s_1, \dots s_m$ if, and only if, all following points hold
\begin{itemize}
\item $t_1 , \dots , t_n$ is produced by $r_1 , \dots , r_{n-1} $ as specified in Definition \ref{def:prod-seq-static-red}
\item $s_1, \dots s_m$ is produced by $r'_1 , \dots , r'_{m-1} $as specified in Definition \ref{def:prod-seq-static-red}
\item the fourth element of none of the quadruples $r_1 , \dots , r_{n-1} ,r'_1 , \dots , r'_{m-1}$ is $\omega$
\item there is an $i$ such that 
\begin{itemize}
\item $1\leq i\leq \min (n-1, m-1)$,
\item $r_j=r'_j$ for $j<i$,
\item $r_i\neq r'_i$, and
\item either $r_i=(\calc[u],\calc[v],q,\textit{left})$ and $r'_i=(\calc[u],\calc[v'],1-q,\textit{right})$,\\or $r_i=(\calc[u],\calc[v'],1-q,\textit{right})$ and $r'_i=(\calc[u],\calc[v],q,\textit{left})$
\end{itemize}
\end{itemize}
\end{definition}

The condition on the rule for adding up the probabilities of the reductions of the same term that yield the same output guarantee that different  $\kappa _i, \kappa _j \in \{ \kappa _1 , \dots , \kappa_n\}$ differ because have been produced by a different nondeterministic choice. In other words, it cannot be that we sum the probabilities of a pair $\kappa _i, \kappa _j $ that have been produced by the same nondeterministic choices. If we enforce any strategy that does not leave any freedom of choice when reducing a deterministic redex, then this condition boils down to requiring that $\kappa_i \neq \kappa _j$ when $i\neq j$.

We can now introduce the rules for the evaluation predicate $\Mapsto$. They are presented in Table \ref{tab:eval-pred-rules}.

The first six rules enable us to construct and extend static reduction sequences. The first one enables us to conclude that any program $t$ evaluates to itself with probability $1$ by the trivial reduction sequence $[t]$. The second rule enables us to extend a reduction sequence by a $\beta$-reduction step, which does not change the probability of the reduction. The third and fourth rule enable us to extend a reduction sequence by a nondeterministic choice step. In this case, the probability of the original reduction sequence is multiplied by the probability of the nondeterminstic choice at hand. The fifth rule is required since a term $t$ might be able to reduce to a term $s$ through different reduction sequences $\kappa _1 ,\dots , \kappa_n$. The rule enables us to represent by $[t,[\kappa _1 /\dots / \kappa_n], s]$ the different ways in which $t$ can reduce to $s$ and suitably compute the overall probability that this happens. The four remaining rules enable us to handle the evaluation of oracles. Two rules handle the actual reduction and two handle the compute of probabilities.

\begin{figure*}[t]\centering

\[\vcenter{
\infer{\Gamma \vdash\,\Mapsto\,  : (\Pi x:A.\Pi y:A. \ast)}{\Gamma\vdash \Pi x:A.\Pi y:A. \ast  :\square  }}  \qquad 
\vcenter{\infer{\Gamma  \vdash [ t]: t  \Mapsto ^1 t}{\Gamma \vdash t:A}}\]

\[\infer{\Gamma \vdash [t, \kappa , \mathcal{C}[(\lambda x:A .t )s] , \mathcal{C}[t[s/x]]]: t \Mapsto ^p \mathcal{C}[t[s/x]]}{\Gamma \vdash [ t, \kappa , \mathcal{C}[(\lambda x:A .t )s] ] :t \Mapsto ^p  \mathcal{C}[(\lambda x:A .t )s]}\]

\[\infer{\Gamma \vdash [t, \kappa , \mathcal{C}[(t\cho{p}s)\nu], \mathcal{C}[t]] :t \Mapsto ^{q\cdot p} \mathcal{C}[t]}{\Gamma \vdash [t, \kappa ,  \mathcal{C}[(t\cho{p}s)\nu]] : t \Mapsto ^q \mathcal{C}[(t\cho{p}s)\nu]} 
\qquad 
\infer{\Gamma \vdash  [t, \kappa , \mathcal{C}[(t\cho{p}s)\nu], \mathcal{C}[s]]:t \Mapsto ^{q\cdot (1-p)}  \mathcal{C}[s]}{\Gamma \vdash [t, \kappa , \mathcal{C}[(t\cho{p}s)\nu]] :t \Mapsto ^q  \mathcal{C}[(t\cho{p}s)\nu]}\]

\[\infer{\Gamma \vdash [t, [\kappa _1 / \dots /\kappa_n] , s]:t\Mapsto_{p_1+\dots +p_n} s}{\Gamma \vdash [t, \kappa _1, s] :t\Mapsto_{p_1} s &\dots & \Gamma \vdash [t, \kappa _n,s] :t\Mapsto_{p_n}s} 
\]where, for any $\kappa _i, \kappa_j\in \{\kappa _1 , \dots , \kappa_n\}$, if $i\neq j$ then $\kappa _i \not\equiv_{\mathrm{ND}} \kappa_j$

\[ \infer{\Gamma \vdash [ \lan o_0\nu, \dots , o_0\nu \ran , \lan s_1 , \dots , s_n \ran] : \lan o_0\nu, \dots , o_0\nu \ran \Mapsto ^1 \lan s_1 , \dots , s_n \ran}{\Gamma \vdash  [ \lan o_0\nu, \dots , o_0\nu \ran] : \lan o_0\nu, \dots , o_0\nu \ran\Mapsto ^1 \lan o_0\nu, \dots , o_0\nu \ran } 
\]\[
\infer{\Gamma \vdash [\lan (o_1t_1)\nu, \dots , (o_1t_n)\nu \ran ,  \lan u_1 , \dots , u_n \ran] :\lan (o_1t_1)\nu, \dots , (o_1t_n)\nu \ran \Mapsto ^1 \lan u_1 , \dots , u_n \ran}{\Gamma \vdash [ \lan (o_1t_1)\nu, \dots , (o_1t_n)\nu \ran ] : \lan (o_1t_1)\nu, \dots , (o_1t_n)\nu \ran\Mapsto ^1 \lan (o_1t_1)\nu, \dots , (o_1t_n)\nu \ran }\]

where $o_0$ is a $0$-ary oracle constant; $\lan o_0\nu, \dots , o_0\nu \ran= \calc  [o_0]_1 \dots  [o_0]_n$; $f_{o_0}(\calc  [\;]_1 \dots  [\;]_n , i )=s_i$ for $1\leq i\leq n$; $o_1$ is a $1$-ary oracle constant;  $\lan (o_1t_1)\nu , \dots , (o_1t_n)\nu\ran = \cald[o_1]_1\dots[o_1]_n$; and $f_{o_1}(\cald[\;]_1\dots[\;]_n, i )=u_i$ for $1\leq i\leq n$

%\[\infer{\Gamma \vdash [\kappa , 
%\mathcal{C}[\lan ot, \dots , ot \ran]]^{q\cdot f}:\mathcal{C}[\lan s_1 , \dots , s_n \ran ]}{\Gamma \vdash [\kappa]^q :\mathcal{C}[\lan ot, \dots , ot \ran ]}\]
%{\latent {\bf BUT HOW?} where $f$ is the\dots}
%
%\[\infer{\Gamma \vdash [\kappa , 
%\mathcal{C}[\lan o, \dots , o \ran]]^{q\cdot f}:\mathcal{C}[\lan s_1 , \dots , s_n \ran ]}{\Gamma \vdash [\kappa]^q :\mathcal{C}[\lan o, \dots , o \ran ]}\]
%{\latent {\bf BUT HOW?} where $f$ is the\dots}

%\[ \infer{\Gamma \vdash [\kappa, ts]^{p\cdot q} :ts'}{\Gamma \vdash [\kappa]^p :ts  & \Gamma\vdash [\kappa ' ]^q: s\mapsto _q s'  }\]

%\[\infer{\Gamma \vdash [t, \kappa , s]: t\Mapsto_p s}{\Gamma \vdash [t, \kappa]^p:\downarrow s}\]

\[\infer{\Gamma \vdash [\lan (ot)\nu, \dots , (ot)\nu \ran , \lan s_1, \dots , s_n \ran ]: ot\Mapsto ^{m/n} s_i}{\Gamma \vdash  [\lan (ot)\nu, \dots ,  (ot)\nu \ran ,  \lan s_1, \dots , s_n \ran ]:\lan  (ot)\nu, \dots ,  (ot)\nu \ran  \Mapsto ^1 \lan s_1, \dots , s_n \ran}\]\[ \infer{\Gamma \vdash [\lan o\nu, \dots , o\nu \ran , \lan s_1, \dots , s_n \ran ]: o\nu\Mapsto ^ {m/n} s_i}{\Gamma \vdash  [\lan o\nu, \dots , o\nu \ran , \lan s_1, \dots , s_n \ran]:\lan o\nu, \dots , o\nu \ran\Mapsto ^1 \lan s_1, \dots , s_n \ran}
\]where $m$ is the number of occurrences of $s_i$ in $s_1, \dots , s_n$

%and $\langle \dots \rangle $ denotes a series of pairs nested on the left

%where, for each $i,j\in \{1,\dots,n\}$, $\kappa_i \neq \kappa _j$ when $i\neq j$ and, moreover,
%
%the first element $e$ that follows the longest prefix that $\kappa_i$ shares with $\kappa _j$ 
%
%must be the immediate result of a nondeterministic choice reduction

%{\latent THIS IS USELESS, LET US REMOVE IT: \[\infer{\Gamma \vdash t^p u: \downarrow v}{\Gamma \vdash t : u\Mapsto_p v & \Gamma \vdash u:A} \]}

\caption{Rules for the  Evaluation Predicate $\Mapsto $}
\label{tab:eval-pred-rules}
\end{figure*}

%
%\section{Trust}
%\label{sec:trust}

%Fixed a target distribution $f$ describing the desired behaviour of all programs in probabilistic terms, a definition of trust with respect to a pair consisting of a program   (the one that will instantiate the variable $x$) and an output  (the one that will instantiate the variable $y$)  is the following:
%\[\all x.\all y. (\neg  \all  z. \neg (x\Mapsto _z y \et \all k . (x\Mapsto _k y\impl k<z) \et  \mid (fxy)-z\mid <fx\varepsilon ) \impl \trust (x, y, f))\]
%where $f$ is the function that computes, for each term $x$, the probability with which $x$ should return $y$. In other words, $fx$ is the target probability function according to which $x$ should compute its outputs.

We can now introduce the trustworthiness predicate $\trust$ and the relative constant term $\tconst$.

Fixed a target distribution $f$ describing the desired behaviour of all programs in probabilistic terms, a definition of trust with respect to a program  (the one that will instantiate the variable $x$) is the following:
\[\all x:A. (\all y:A.( f x y \neq 0 \impl \neg  \all  z:\boldq. \neg ( x \Mapsto _{z} y\wedge   \vert (fxy)-z\vert <fx\varepsilon))\impl \trust (x, f ))\]where $f: A\impl A\impl\boldq $ is a function that computes, for each term $x:A$, the probability with which $x$ should return $y$. In other words, $fx:A\impl \boldq$ is the probability function according to which $x$ should compute its outputs, $fx$  describes the desired behaviour of $x$. The application $fx\thr$, on the other hand, returns a threshold (a rational number between $0$ and $1$) that we use to check the trustworthiness of $x$.

This definition is added to the system as the following axiom:
\[\tconst :\all x:A. (\all y:A.( f x y \neq 0 \impl \neg  \all  z:\boldq. \neg ( x \Mapsto _{z} y\wedge   \vert (fxy)-z\vert <fx\varepsilon))\impl \trust (x, f ))\]where $\tconst$ is a constant term.

The system also need to include the following rule for introducing suitable sentences universally quantifying over the set of outputs of a term $t$ that have non-zero probability of being returned by $t$: 
\[\infer{\Gamma \vdash \bolds (s_0, s_1, \dots , s_n, s'_1, \dots , s'_n): \all y. (fty \neq 0 \impl A)}{\Gamma \vdash s_0:m_1 +\dots +m_n=1&\{\Gamma \vdash s_i: t\Mapsto _{m_i} y_i\}_{1\leq i \leq n}& \{ \Gamma \vdash s'_j: A[y_j/y]\}_{1\leq j \leq n}}\]
where the notation $\{\Gamma \vdash t_i:B_i \}_{1\leq i \leq n}$ indicates the list of premisses $\Gamma \vdash t_1:B_1 \;\;\dots \;\; \Gamma \vdash t_n:A_n $. 

In order to prove the left-hand side of the outermost implication of the definition of $\trust$, we can use an application of this rule in which $A=\neg  \all  z. \neg (t\Mapsto _z y \et \vert (fty)-z\vert <fx\varepsilon$.

%Finally, in case we need to avoid the use of conjunction (if we also introduce a propositional 2nd order universal quantifier, we can define it, but plain dependent types do not enable us to define conjunction by $\all$, I think), we can treat the conjunction of the two atomic formulae in the previous definition of $\trust$ as one atomic formula:\[\infer{\Gamma \vdash (s,t): u \Mapsto _{m} v\wedge   \vert (fuv)-m\vert <fu\thr}{\Gamma \vdash s:u \Mapsto _{m} v& \Gamma \vdash t: \vert (fuv)-m\vert <fu\varepsilon}\]where $u \Mapsto _{m} v\wedge   \vert (fuv)-m\vert <fu\thr$ should be read as an unusual notation for an atomic formula in which the terms $u,m,v, f, \thr$ occur.\footnote{We simply use the technical trick of keeping the notation  $u \Mapsto _{m} v\wedge   \vert (fuv)-m\vert <fu\thr$ for the atomic formula that we would usually denote by an atom of the form $P(u,m,v,f,\thr )$. We do this in order to avoid the necessity to include conjunction rules in the calculus.}
%
%\section{Examples}
%\label{sec:examples}
%
%{\ehi If we want to check fairness, trust or similar things, concerning the behaviour of a classifier on a population, let us check trust on the program defined as an application of the classifier $o$  to a randomised selector of an input individual among $\mathtt{input}_1 , \dots , \mathtt{input}_n$:
%\[o((\mathtt{input}_1 \cho{\frac{1}{n}} (\mathtt{input}_2 \cho{\frac{1}{n-1}}( \mathtt{input}_3  \cho{\frac{1}{n-2}}( \;\dots\; \cho{\frac{1}{2}} \mathtt{input}_n)\nu)\nu)\nu)\nu \dots)\nu)\]
%}
% 

\bibliographystyle{apalike}
\bibliography{bib-dependent-types-trust.bib}

\end{document}